\documentclass[12pt]{extarticle}
\usepackage[latin9]{inputenc}
\usepackage{geometry}
\usepackage{verbatim}
\usepackage{float}
\usepackage{amsmath}
\usepackage{amsthm}
\usepackage{amssymb}
\usepackage{graphicx}
\usepackage{rotfloat}
\usepackage{setspace}
\usepackage[authoryear]{natbib}
\usepackage[unicode=true,
 bookmarks=false,
 breaklinks=false,pdfborder={0 0 1},backref=false,colorlinks=false]
 {hyperref}

\makeatletter

\providecommand{\tabularnewline}{\\}

\numberwithin{equation}{section}
\numberwithin{figure}{section}
\theoremstyle{plain}
\newtheorem{assumption}{\protect\assumptionname}
\theoremstyle{plain}
\newtheorem{lem}{\protect\lemmaname}
\theoremstyle{plain}
\newtheorem{thm}{\protect\theoremname}
\theoremstyle{remark}
\newtheorem{rem}{\protect\remarkname}

\makeatother

\providecommand{\assumptionname}{Assumption}
\providecommand{\lemmaname}{Lemma}
\providecommand{\remarkname}{Remark}
\providecommand{\theoremname}{Theorem}

\begin{document}
\title{\vspace{-8ex}Identification of Time-Varying Transformation Models
with Fixed Effects, with an Application to Unobserved Heterogeneity
in Resource Shares}
\author{Irene Botosaru, Chris Muris, and Krishna Pendakur\vspace{-4ex}\thanks{This paper partly supersedes and replaces an earlier working paper
by Botosaru and Muris entitled \textquotedblleft Binarization for
panel models with fixed effects\textquotedblright . Email: botosari@mcmaster.ca,
muerisc@mcmaster.ca, pendakur@sfu.ca. We would like to thank Stéphane
Bonhomme, P.A. Chiappori, Ian Crawford, Iván Fernández-Val, Dalia
Ghanem, Bryan Graham, Hiro Kasahara, Shakeeb Khan, Valerie Lechene,
Arthur Lewbel, Corrine Low, Maurizio Mazzocco, Jack Porter and Elie
Tamer for useful discussions and suggestions; conference participants
at the ASSA 2018, Berkeley-Stanford Jamboree 2017, IWEEE 2018, Seattle-Vancouver
Econometrics Conference 2016 and 2017, and seminar participants at
Cornell, CREST, Harvard, Oxford, Penn State, Queen's University, Rotterdam,
Tilburg, University of Bristol, TSE, UBC, UCL, UC Davis, UCLA, UPenn
and Vanderbilt for comments and suggestions. We thank Shirleen Manzur
for her research assistance on this project, and gratefully acknowledge
financial support from the Social Sciences and Humanities Research
Council of Canada under grants IDG 430-2015-00073, IG 435-2020-0026,
and IG 435-2021-0778, and from the Institute for Advanced Study at
the University of Bristol. This research was undertaken, in part,
thanks to funding from the Canada Research Chairs Program.\protect }\vspace{-8ex}}
\maketitle
\begin{abstract}
We provide new results showing identification of a large class of
fixed-$T$ panel models, where the response variable is an unknown,
weakly monotone, time-varying transformation of a latent linear index
of fixed effects, regressors, and an error term drawn from an unknown
stationary distribution. Our results identify the transformation,
the coefficient on regressors, and features of the distribution of
the fixed effects. We then develop a full-commitment intertemporal
collective household model, where the implied quantity demand equations
are time-varying functions of a linear index. The fixed effects in
this index equal logged \textit{resource shares}, defined as the fractions
of household expenditure enjoyed by each household member. Using Bangladeshi
data, we show that women's resource shares decline with household
budgets and that half of the variation in women's resource shares
is due to unobserved household-level heterogeneity. 
\end{abstract}

\section{Introduction\label{sec:introduction}}

\begin{comment}
There is a vast literature on static fixed$-T$ panel models where
the response variable is a nonlinear transformation of a linear latent
index of unobserved individual-specific effects and observed explanatory
variables. Many papers in this literature focus on showing (point)
identification of the coefficients on the explanatory variables and,
more recently, of the transformation of the latent index. These papers
make various assumptions on: the individual-specific effects (fixed
versus correlated random effects); the transformation (e.g., time-invariant
versus time-varying, and strictly- versus weakly-monotone); and, the
distribution of the error terms (e.g., parametric versus non-parametric).
In this paper, we provide a systematic approach to show the identification
of a general class of static, fixed-$T$, nonlinear panel models,
with fixed effects, weakly monotone and time-varying transformation,
and error terms that are drawn from a distribution that may be nonparametric.
\end{comment}
{} 

We provide sufficient conditions for point-identification in a general
class of fixed-$T$ time-varying nonlinear panel models. This class
has a response variable equal to a time-varying weakly monotonic transformation
of a linear index of regressors, fixed effects, and error terms. In
contrast, almost all existing results for this class of models require
time-invariance of the transformation. Our theorems imply novel identification
results for time-varying versions of some commonly used models. 

Specifically, we consider models of the type:
\begin{equation}
Y_{it}=h_{t}(\alpha_{i}+X_{it}\beta-U_{it}),i=1,\dots,n,\,t=1,\dots,T\geq2,\label{eq:FELT}
\end{equation}
where $h_{t}$ is an unknown weakly monotonic transformation, $\alpha_{i}$
are unobserved individual-specific effects, $X_{it}$ is a vector
of strictly exogenous observed explanatory variables with coefficients
$\beta$, and $U_{it}$ is an error term drawn from a stationary distribution. 

Our setting has the following four features:
\begin{enumerate}
\item fixed-$T$ ---in fact $T=2$ is sufficient for our results;
\item $h_{t}$ may be \textbf{time-varying} and \textbf{weakly} monotonic;
\item $h_{t}$ and the distribution of $U_{it}$ may be nonparametric;
\item and, the fixed effects $\alpha_{i}$ are \textbf{unrestricted.}
\end{enumerate}
The panel model literature is very large, with many papers showing
identification of $\beta$, and sometimes of $h_{t}$, in models with
two or three of the four features above (see, e.g., \citet{ArellanoBonhomme2012}
for an overview). Ours is the first paper to study identification
of models with these four features, all of which are demanded by our
microeconomic model and empirical application. We refer to (\ref{eq:FELT})
together with the four features above as the \emph{fixed-effects linear
transformation} (FELT) model, following the terminology in \citet{Abrevaya1999}.%
\begin{comment}
This setting excludes some important types of models: dynamic models
(e.g. \citet{HonoreKyriazidou2000Ecta}, \citet{Aguiregabiria2020},
\citet{KhanPonoTamer2020}); and models of multinomial choice (e.g.,
\citet{ShiShumSong2016}).
\end{comment}

Our contribution is at least three-fold. First, we provide sufficient
conditions for the identification of $h_{t}$ and $\beta$ for models
in the FELT class. Second, for the case where $h_{t}$ is strictly
monotonic, we provide results on the identification of some features
of the distribution of fixed effects. Third, we provide a full-commitment
intertemporal collective household model whose implied quantity demand
equations lie in the FELT class, and estimate the model using Bangladeshi
panel data.

We provide sufficient conditions for the point-identification of $h_{t}$
and $\beta$ in (\ref{eq:FELT}) for two non-nested cases: one where
$U_{it}$ is drawn from an arbitrary but stationary distribution,
and one where it is drawn from the logistic distribution. Our approach
provides a systematic way of analyzing models nested in the FELT class.
An immediate implication of our work is that extensions to time-varying
and/or nonparametric counterparts of well-known models can now be
shown to be identified. For example, the following models are all
nested in the FELT class and are now identified: ordered choice with
time-varying thresholds; censored regression with time-varying censoring
points; the multiple-spell generalized accelerated failure-time (GAFT)
duration model; and the Box-Cox panel model with time-varying parameters.%
\begin{comment}
Our reading of the literature, which we summarize in Section \ref{sec:lit-review},
is that \textit{arbitrary} time-variance of transformations in the
panel context is not to be taken for granted.
\end{comment}

For the case where $h_{t}$ is strictly monotonic we provide additional
identification results for the conditional mean (up to location) and
the conditional variance of the distribution of fixed effects. To
the best of our knowledge, there are no results in the fixed-$T$,
fixed-effects, nonlinear panel literature that cover this aspect of
model identification. 

Our theoretical work builds on established results from \citet{DoksumGasko1990}
and \citet{Chen2002} who show that cross-sectional transformation
models can in general be \emph{binarized} into a set of related binary
choice models. In this paper, we show that we can binarize in a panel
setting when the transformation $h_{t}$ varies in an arbitrary way
over time. 

The key innovation underlying our theoretical work is \emph{time-varying
}binarization. For an arbitrary threshold $y_{t}$, we define the
following binary random variable:
\begin{align}
D_{t}\left(y_{t}\right) & \equiv1\left\{ Y_{t}\geq y_{t}\right\} \label{eq:defineDt-1}\\
 & =1\left\{ U_{t}\leq\alpha+X_{t}\beta-h_{t}^{-}\left(y_{t}\right)\right\} ,\nonumber 
\end{align}
where $h_{t}^{-}$ is the generalized inverse of $h_{t}$ and where
the equality follows from specification \eqref{eq:FELT} and weak
monotonicity of $h_{t}$. Varying the threshold $y_{t}$ in (\ref{eq:defineDt-1})
across time periods converts a FELT model into a collection of binary
choice models. This conversion is what we call time-varying binarization.\footnote{\citet{Muris} uses time-varying binarization in a panel ordered logit
model where the transformation is time-invariant and parametric.} However, to the best of our knowledge, previous papers that used
binarization in a panel setting, e.g., \citet{Chen2010}, \citet{ChernozhukovWP2018},
restrict the thresholds to be equal across time periods, i.e. $y_{t}=y$
for all $t$. Essentially, it is the relaxation of this restriction
that enables us to show identification of time-varying transformations
$h_{t}$.

Once a FELT model has been converted into a collection of binary choice
models via time-varying binarization, we invoke \citet{Chamberlain1980}
and \citet{Manski1987} to show identification of the resulting binary
choice models. We then re-assemble the identified models to obtain
identification of $h_{t}$ and $\beta$ in the FELT model. Omitting
the fact that any FELT model can be transformed into \textit{many}
binary choice models obtains identification of $\beta$ only. 

We provide a full-commitment intertemporal collective household model
that implies time-varying quantity demand functions of the form (\ref{eq:FELT}).
The nonlinear quantity demand functions in our model are time-varying
because quantity demands depend on prices, and prices are unobserved
but vary across the waves of our panel. The fixed effects in our economic
model have a clear interpretation: they are logged \textit{resource
shares}, defined as the fractions of total household expenditure consumed
by each of its members. Resource shares are not directly observable,
but are important because unequal resource shares across household
members signal within-household inequality. Our econometric results
then imply that the quantity demand equations, the mean (up to location),
and the variance of the resource shares are all identified. This is
useful because it allows us to characterize the variation of, or inequality
in, resource shares. 

Our intertemporal collective household model ---along with the identification
results above--- permits the use of short panel data to study resource
shares within households. Previous cross-sectional methods to measure
resource shares have imposed the identifying restriction that resource
shares do not vary with household budgets, e.g., \citet{dlp13}, and
have relied on a random-effects model for unobserved heterogeneity
in resource shares, e.g., \citet{Dunbarlp19}. Our framework relaxes
both these restrictions, and shows how panel data can enrich the study
of resource shares. Ours are the first empirical estimates of a full-commitment
collective household model in a short panel, and we demonstrate the
importance of accounting for both observed and unobserved heterogeneity
in women\textquoteright s resource shares.

Using a two-period Bangladeshi panel dataset on household expenditures,
we show that less than half of the variation in women's resource shares
can be explained by observed covariates. This means that there is
much more inequality within households than would be suggested by
variation in observed factors. We also find that women's resource
shares are negatively correlated with household budgets. This means
that women in poorer households have larger resource shares, and are
therefore less poor than their household budgets would suggest. Further,
that resource shares are found to covary with household budgets suggests
caution in using the cross-sectional identifying restriction of independence
suggested by \citet{dlp13}.

Section \ref{sec:lit-review} provides a review of the related literature.
In Sections \ref{sec:fixed-effects-linear-transformation-model} and
\ref{sec:FE}, we provide our main identification results. Section
\ref{sec:MicroModel} introduces our collective household model, which
uses data described in Section \ref{sec:Data}. We present estimates
of women's resource shares in Section \ref{sec:ResourceShares}. All
proofs, descriptive statistics for the data, additional robustness
results and estimation details are in the Appendix.

\section{Existing Literature\label{sec:lit-review}}

\begin{comment}
FELT includes many previously studied models. For example, when the
transformation is \emph{weakly} monotonic, it nests, e.g., binary
choice, ordered choice, and censored models. FELT nests binary choice
models with time effects as a special case when $h_{t}(Y_{it}^{*})=1\left\{ Y_{it}^{*}\geq\lambda_{t}\right\} $.
In these models, $\beta$ is known to be identified for $U_{it}$
parametric or nonparametric, see, e.g., \citet{Rasch1960}, \citet{Chamberlain1980},
\citet{Manski1987}, \citet{Magnac2004}, and \citet{Chamberlain2010}.
When the transformation is strictly monotonic, it nests, e.g., duration
and Box-Cox regression models. Since the transformation in FELT can
be time-varying, our framework generalizes typical discrete- and continuous-choice
models where $h_{t}$ is fixed over time, providing identification
results for models that were not previously shown to be identified. 
\end{comment}

We describe in detail how we connect to, first, the econometrics literature,
and, second, to the literature on collective household models.

\subsection{Fixed-$T$ Nonlinear Panel Models with Fixed Effects}

The literature on panel models is vast. Despite the vastness of this
literature, we are not aware of any paper that delivers all four features
discussed above. Below, we highlight the key differences between our
approach and approaches in the literature that \textit{lack} one or
more of our key features.

\textbf{Feature 1: We show identification in fixed-$T$ panel models}.
The incidental parameter problem occurs in fixed-effect panel models
with a finite number of time periods, see \citet{NeymanScott1948}.%
\begin{comment}
Essentially, the problem arises from the fact that the $n$ fixed
effects $\alpha_{i}$ cannot be consistently estimated if $T$ does
not tend to infinity. Thus, identification of the parameters common
across individuals must be shown in a context where the incidental
parameters $\alpha_{i}$ cannot be identified or consistently estimated.
\end{comment}
{} Ours is a fixed-$T$ approach, with $n\to\infty$ and works even
if $T=2$. A large literature analyzes the behavior of fixed effects
procedures under the alternative assumption that the number of time
periods goes to infinity, e.g., \citet{HahnNewey}, \citet{ArellanoHahn2007},
\citet{ArellanoBonhomme2009}, \citet{FernandezVal2009}, \citet{FernandezValWeidner2016},
and \citet{ChernozhukovWP2018}. In this setting, it is generally
possible to identify each fixed effect, and consequently, the distribution
of fixed effects. In our model, we show identification of specific
moments of this distribution even though the number of time periods
is fixed.

\textbf{Feature 2: We allow weakly monotonic time-varying transformation}
$h_{t}$. \citet{Abrevaya1999} provides a consistent estimator of
$\beta$ (the ``leapfrog'' estimator) in the FELT model under the
restriction that the transformations are \emph{strictly} monotonic.\footnote{\citet{SChen2010b} in Remark 6 discusses a version of \citet{Abrevaya1999}
that allows for some weak monotonicity due to censoring. He focuses
on $\beta$ and does not discuss identification of $h_{t}$, although
his Remark 1 sketches an approach for estimation of $h_{t}=h$ for
all $t$.} %
\begin{comment}
However, because he differences out the transformations $h_{t}$,
his focus does not extend to their identification.
\end{comment}
{} %
\begin{comment}
\citet{AtheyImbens2006} propose a ``changes in changes'' estimator,
which is a generalization of the linear differences-in-differences
estimator, in both a cross-sectional and a panel data setting. Their
panel data fixed-effects setting is a potential outcomes analog to
our model with strictly monotonic transformations. They show identification
of the average treatment effect, but not identification of $h_{t}$
or the distribution of $\alpha_{i}$. In comparison, we cover the
weakly monotone case, and identify the transformation $h_{t}$. In
the strictly monotone case, we additionally identify moments of the
distribution of fixed effects $\alpha_{i}$. 
\end{comment}
{} \citet{Abrevaya2000} considers a model that allows for \emph{weak}
monotonicity but restricts the transformations to be time-invariant
(and allows for nonseparable errors). He provides a consistent estimator
for $\beta$ only.\footnote{\citet{ChernozhukovWP2018} uses a distribution regression technique
that is closely related to our binarization approach, and consequently
accommodates weakly monotonic transformations. However, theirs is
a large-$T$ setting.} A literature on duration models also considers time-invariant transformations
that are weakly monotonic due to censoring, e.g., \citet{Lee2008},
\citet{Khan2007}, \citet{Chen2010,Chen2010a,Chen2012,Chen2012ET},
and \citet{Chen2012}; we review this below.

A more recent literature has focused on identification issues in a
class of panel models with potentially \emph{non-monotonic} but time-invariant
structural functions (or strong assumptions on how those functions
vary over time), e.g., \citet{HoderleinWhite2012}, \citet{ChernValHahnNewey2013},
\citet{ChernValHoderleinHolzNewey2015}. These papers focus on (partial)
identification of partial effects, and the approaches employed there
preclude identification of the structural function(s) or of the distribution
of fixed effects.

\textbf{Feature 3: We allow nonparametric transformations and nonparametric
errors}. \citet{Bonhomme2012} proposes a general-purpose likelihood-based
approach to obtain identification for models with parametric $h_{t}$
and parametric $U_{it}$, even allowing for dynamics. %
\begin{comment}
These results exploit the fact that a likelihood function can be constructed
for such models and show identification in the presence of fixed effects
in a finite-$T$ setting. (Earlier work for the same setting by \citet{Lancaster2002}
requires $T\to\infty$, sacrificing feature 1.)
\end{comment}
{} Our model requires strictly exogenous regressors, precluding many
dynamic structures. But, our Theorems 1 and 2 apply even when $h_{t}$
is nonparametric, $U_{it}$ is nonparametric, or both are nonparametric. 

The setting with parametric transformations and parametric errors
covers many models previously shown to be identified, including the
\emph{time-invariant} fixed-effect panel versions of: binary choice
(e.g., \citet{Rasch1960}, \citet{Chamberlain1980}, \citet{Magnac2004},
and \citet{Chamberlain2010}); the linear regression model with normal
errors; and the ordered logit model (e.g., \citet{das_panel_1999,Baetschmann2015,Muris}).
Application of our results immediately shows identification of the
time-varying versions of these models. This result is novel for the
ordered logit model, where our results imply identification of time-varying
thresholds.

Parametric transformation models with nonparametric errors are widely
studied, starting with \citet{Manski1987} for the binary choice fixed
effects model. (\citet{Aristodemou2020} provides partial identification
results for ordered choice with nonparametric errors.) Parametric
panel data censored regression models also fit into our framework,
and were studied intensively starting with \citet{Honore1992}, see
also, e.g., \citet{Charlier2000}, \citet{HonoreKyriazidou2000Ecta},
\citet{Chen2012ET}. These papers show identification of $\beta$
for the linear model with time-invariant censoring and nonparametric
errors. In this context, our results show identification of models
that were not previously known to be identified. In particular, the
model is identified even if the transformation is nonparametric (as
opposed to linear or Box-Cox) and time-varying and/or where the censoring
cutoff is time-varying.\footnote{Many papers in the literature on censored regression have focused
on endogeneity. For example, \citet{HonoreHu2004} allow for endogenous
covariates, and \citet{KhanPonomarevaTamer2016} study the case of
endogenous censoring cutoffs. Our results do not cover the case of
endogenous regressors or cutoffs. \citet{HorowitzLee2004} and \citet{Lee2008}
consider dependent censoring, where the censoring cutoff depends on
observed covariates and the error term follows a parametric distribution.
We do not consider dependent censoring.}

Duration models can be recast as transformation models with nonparametric
transformations (see \citet{Ridder1990}). Consequently, the large
literature on identification of duration models is related to our
work. 

Consider the multiple-spell mixed proportional hazards (MPH) model
with spell-specific baseline hazard, analyzed in \citet{Honore1993}.\footnote{\citet{HorowitzLee2004} show identification of this model under the
restriction that the baseline hazard is the same for all spells, analogous
to time-invariant $h_{t}$. \citet{Chen2010a} considers the same
model, but relaxes the restriction that errors are type 1 EV, but
shows identification of only the common parameter vector $\beta$.} This model can be obtained from FELT by letting (i) $h_{t}^{-1}(v)=\log\left\{ \int_{0}^{v}\lambda_{0t}\left(u\right)du\right\} $,
where $\lambda_{0t}$ is the baseline hazard for spell $t$, (ii)
$\alpha_{i}$ and $U_{it}$ are independent across $t$, and (iii)
$U_{it}$ is independent of $X_{i}$ and distributed as EV1. \citet{Honore1993}
derives sufficient conditions for the identification of this model
(\citet{Lee2008} provides consistent estimators under other parametric
error distributions). Our theorems immediately provide the novel result
that this model is identified when the error terms are drawn from
a nonparametric distribution.

Consider the single-spell generalized accelerated failure time (GAFT)
model introduced by \citet{Ridder1990} (see also \citet{vandenBerg2001})
that has non EV1 errors, and is consistent with a duration model.
Just like the MPH model, it can be extended to a multiple-spell setting,
see, e.g., \citet{Evdokimov2011}. \citet{Abrevaya1999} shows that
$\beta$ in the multiple-spell GAFT model is consistently estimated.
However, he does not show identification of the transformation $h_{t}$,
which can be seen as dual to identification of the spell-specific
baseline hazard function.\footnote{\citet{Khan2007} establish consistency of an estimator of the regression
coefficient in GAFT under the restriction that the baseline hazard
is the same for all spells, analogous to time-invariant $h_{t}$. } \citet{Evdokimov2011} considers identification of a related version
of the multiple-spell GAFT with spell-specific baseline hazard, but
requires continuity of $\alpha_{i}$ and at least 3 spells ($T\geq3)$.
Our results show identification of both $\beta$ and $h_{t}$ in the
multiple-spell GAFT model, imposing no restrictions on $\alpha_{i}$
and requiring just 2 spells ($T=2$).

\textbf{Feature 4: We allow for unrestricted fixed effects.} A related
literature considers restrictions on the joint distribution of $\left(\alpha_{i},X_{i1},...,X_{iT}\right)$.
For example, \citet{AltonjiMatzkin2005} impose exchangeability on
this joint distribution, and \citet{BesterHansen2009} restricts the
dependence of $\alpha_{i}$ on $\left(X_{i1},...,X_{iT}\right)$ to
be finite-dimensional. In our model this joint distribution is unrestricted.

A further group of papers establishes identification of panel models,
including the distribution of $\alpha_{i}$, by using techniques from
the measurement error literature that: (i) impose various assumptions
on $\alpha_{i}$, such as full support and/or continuous distribution;
(ii) assume serial independence of $U_{it}$; and (iii) restrict the
conditional distribution of $\left(\alpha_{i},U_{i1},\dots,U_{iT}\right)$
conditional on $\left(X_{i1},...,X_{iT}\right)$, see, \citet{Evdokimov2010},
\citet{Evdokimov2011}, \citet{Wilhelm2015}, and \citet{Freyberger18}.
In contrast, our results on the identification of the conditional
variance of $\alpha_{i}$ do not require (i). All our other results,
including identification of the dependence of $\alpha_{i}$ on observed
covariates, are free of assumptions like (i), (ii) and (iii).

Finally, special regressor approaches (see the review in \citet{Lewbel2014})
have identifying power in transformation models with fixed effects.
They require the availability of a continuous variable that is independent
of the fixed effects. With such a variable, one can show identification
of transformation models in the cross-sectional case (\citet{ChiapporiKomunjerKristensen2015})
and in the panel data case, e.g., \citet{HonoreLewbel2002}, \citet{AiGan2010},
\citet{lewbelyang2016}, \citet{chen_exclusion_2019}. Our results
do not invoke a special regressor. Further, we are not aware of any
special regressor-based papers that identify time-varying transformations
or the distribution of fixed effects.\footnote{We conjecture that the existence of a special regressor would be sufficient
to identify time-varying nonparametric transformations, and, with
strict monotonicity, the distribution of fixed effects. However, we
think that a setting with completely unrestricted fixed effects is
useful in a variety of empirical applications, including our own.} 

\textbf{We also show identification of some aspects of the distribution
of fixed effects.} Correlated random effects models identify the distribution
of individual effects, but at the cost of restricting their distribution.
To our knowledge, we are the first to show identification of moments
of this distribution in a nonlinear panel model, when that distribution
is unrestricted. 

We show the practical importance of these innovations in our empirical
work below. Identification of the conditional mean and variance of
the distribution of fixed effects in a context with time-varying transformations
is essential to our investigation of women's access to household resources
in rural Bangladeshi.

\subsection{Microeconomic Models of Collective Households}

Dating back at least to \citet{becker62}, collective household models
are those in which the household is characterized as a collection
of individuals, each of whom has a well-defined objective function,
and who interact to generate household level decisions such as consumption
expenditures. \emph{Efficient collective household} models are those
in which the individuals in the household are assumed to reach the
(household) Pareto frontier. \citet{Chiappori88,Chiappori92} showed
that, like in earlier results in general equilibrium theory, the assumption
of Pareto efficiency is very strong. Essentially, it implies that
the household can be seen as maximizing a weighted sum of individual
utilities, where the weights are called \emph{Pareto weights}. This
in turn implies that the household-level allocation problem is observationally
equivalent to a decentralized, person-level, allocation problem.

In this decentralized allocation, each household member is assigned
a \emph{shadow budget}. They then demand a vector of consumption quantities
given their preferences and their personal shadow budget, and the
household purchases the sum of these demanded quantities (adjusted
for shareability/economies of scale and for public goods within the
household). For the special case of an assignable good, which is demanded
by a single known household member, the household purchases exactly
what that person demands given their shadow budget.

\emph{Resource shares}, defined as the ratio of each person's shadow
budget to the overall household budget, are useful measures of individual
consumption expenditures. If there is intra-household inequality,
these resource shares would be unequal. Consequently, standard per-capita
calculations (assigning equal resource shares to all household members)
would yield invalid measures of individual consumption and poverty
(see, e.g., \citet{dlp13}). In this paper, we show identification
of the conditional mean (up to location) and conditional variance
of the distribution of resource shares in a panel data context.%

There are many ways to identify resource shares with cross-sectional
data. A common identifying assumption (used by, e.g., \citet{dlp13})
is that resource shares are independent of household budgets in a
cross-sectional sense. This identifying restriction has been used
to estimate resource shares, within-household inequality and individual-level
poverty in many countries (\citet{dlp13} and \citet{Dunbarlp19}
in Malawi; \citet{Bargain2014} in Cote d'Ivoire; \citet{Calvi} in
India; \citet{DeVreyer2016} in Senegal; \citet{Bargain2018} in Bangladesh).
In our model, we show identification of the response of the conditional
mean of resource shares to observed covariates, even if resource shares
are correlated with (lifetime) household budgets. Consequently, we
can test this identifying restriction. 

\citet{dlp13} does not accommodate unobserved heterogeneity in resource
shares. Two newer papers, \citet{cKim17} and \citet{DunbarLewbelPendakur2019}
consider identification in cross-sectional data with unobserved household-level
heterogeneity in resource shares. Like \citet{cKim17} and Theorem
1 in \citet{DunbarLewbelPendakur2019}, our work investigates identification
of the distribution of resource shares up to an unknown normalization.
However, the results in those papers are of the random effects type.
That is, the authors impose the restriction that the conditional distribution
of resource shares is independent of the household budget. In this
paper, we consider a panel data setting with household-level unobserved
heterogeneity in resource shares, without any restriction on the distribution
of resource shares. Further, we show sufficient conditions for identification
of the conditional variance of (logged) resource shares.

The literature cited above considered one-period micro-economic models.
But, many interesting questions about households, and the distribution
of resources within households, are dynamic in nature. For example:
how do household members share risk?; how do household investments
relate to individual consumption?; how can we use information from
multiple time periods to estimate resource shares when there is unobserved
household-level heterogeneity?

\citet{ChiapporiMazzocco2017} review the literature on collective
household models in an intertemporal setting. These models generally
come in two flavours--limited commitment or full commitment--depending
on whether or not the household can commit to a \emph{permanent} Pareto
weight at the moment of household formation. Full-commitment models
answer \textquotedblleft yes\textquotedblright , and limited-commitment
model answer \textquotedblleft no\textquotedblright . Limited commitment
models have commanded the most theoretical attention. Much effort
has gone into testing the full-commitment model against a limited-commitment
alternative, e.g., \citet{Ligon98,Mazzocco07,MazzocoRuizYamaguchi14,Voena15}.

Fewer papers study the identification of Pareto weights or resource
shares in an intertemporal context. \citet{LiseYamada19} use a long
panel of Japanese household consumption data to estimate how Pareto
weights (which are dual to resource shares) depend on observed covariates
and on unanticipated shocks. Their model does not allow for correlated
unobserved heterogeneity, and they require many observations for each
household (so as to see when Pareto weights change). They find evidence
that Pareto weights do change, so that the full-commitment model does
not hold in Japan. 

Our model is one of full-commitment, allows for correlated unobserved
household-level heterogeneity, and is identified in a short (e.g.,
2 period) panel. So we provide a complement to the approach of \citet{LiseYamada19}
for cases where the data are not rich enough to estimate a limited-commitment
model. The key cases here are where the panel is too short to see
changes in Pareto weights, or where the observed covariates leave
too much room for unobserved heterogeneity. The cost of covering these
cases is the assumption of full commitment.

Full commitment models are more restrictive, but may be useful nonetheless.
\citet{ChiapporiMazzocco2017} write \textquotedblleft In more traditional
environments (such as rural societies in many developing countries),
renegotiation may be less frequent since the cost of divorce is relatively
high, threats of ending a marriage are therefore less credible, and
noncooperation is less appealing since households members are bound
to spend a lifetime together.\textquotedblright{} We use a full commitment
setting to estimate resource shares for rural Bangladeshi households.

In this paper, we adapt the general full-commitment framework of \citet{ChiapporiMazzocco2017}
to the scale economy and sharing model of \citet{BrowningChiapporiLewbel13}.
Then, like \citet{dlp13} do in their cross-sectional analysis, we
identify resource shares on the basis of household-level demand functions
for assignable goods%
\begin{comment}
, defined as goods where the expenditure level is observed for each
household member
\end{comment}
. In our general model, observed household-level quantity demand functions
for assignable goods depend on resource shares, and resource shares
depend a time-invariant factor (a fixed effect) representing the initial
(and permanent) Pareto weights of household members.

We then provide a parametric form for utility functions that results
in demand equations for assignable goods that are nonlinear in shadow
budgets, and have logged shadow budgets that are linear in logged
household budgets and a fixed effect. Further, demand equations are
time-varying because prices vary over time. Such demand equations
fall into the FELT class, and are therefore identified in our short-panel
setting. The parametric model also gives meaning to the fixed effect:
it equals a logged resource share, so its distribution is an economically
interesting object. So, our micro-economic theory demands an econometric
model that allows for time-varying transformations and that can identify
moments of the conditional distribution of fixed effects.

\section{Identification\label{sec:fixed-effects-linear-transformation-model}}

Dropping the $i$ subscript, let $Y=\left(Y_{1},...,Y_{T}\right)^{'}$
and $X=\left(X_{1}',...,X_{T}'\right)^{'}$. We then rewrite FELT
as a latent variable model:
\begin{equation}
\begin{aligned}Y_{t} & =h_{t}(Y_{t}^{*})=h_{t}\left(\alpha+X_{t}\beta-U_{t}\right),\\
U_{t} & |\alpha,X\sim F_{t}(u|\alpha,X),
\end{aligned}
\label{eq:model}
\end{equation}
and denote the supports of $Y_{t},\,Y_{t}^{*},\,X_{t}$ by $\mathcal{Y}\subseteq\mathbb{R}$,
$\mathcal{\mathcal{\mathcal{Y}}^{\textnormal{*}}=\mathbb{R}},$ and
$\mathcal{X}\subseteq\mathbb{R}^{K}$, respectively.\footnote{The supports may be indexed by $t$.}

We provide sufficient conditions for identification of $\left(\beta,h_{t}\right)$.\footnote{The results in this section were previously circulated in the working
paper \citet{BotosaruMuris}. That paper also introduces four estimators,
depending on whether the outcome variable is discrete or continuous,
and on whether the stationary distribution of the error term is nonparametric
or logistic. In this paper, we use a GMM estimator instead.} We consider two non-nested cases. The first case does not impose
parametric restrictions on the distribution of $U_{t}$, requiring
only that it is conditionally stationary. In this case, the idiosyncratic
errors may be serially dependent and heteroskedastic. The second case
assumes that $U_{t},\,t=1,\cdots,T$, are serially independent, standard
logistic, and strictly exogenous. %
\begin{comment}
The second case has an error distribution that is nested within that
of the first case.
\end{comment}
{} It may appear that the second case is a special case of the first.
However, the second case requires weaker assumptions on the distribution
of the regressors (c.f. Assumption \ref{A3} below) while imposing
stronger assumptions on the error distribution. For both cases, we
maintain the assumption below:
\begin{assumption}
\label{A1}{[}Weak monotonicity{]} For each $t$, the transformation
$h_{t}:\mathcal{\mathcal{Y^{\mbox{*}}}}\rightarrow\mathcal{Y}$ is
unknown, non-decreasing, right continuous, and non-constant.
\end{assumption}
This assumption allows us to work with the generalized inverse $h_{t}^{-}:\mathcal{Y}\rightarrow\mathcal{Y}^{*}$,
defined as:
\[
h_{t}^{-}\left(y\right)\equiv\inf\left\{ y^{*}\in\mathcal{\mathcal{Y}^{\textnormal{*}}}:\mbox{ }y\leq h_{t}\left(y^{*}\right)\right\} ,
\]
with the convention that $\text{inf}\left(\emptyset\right)=\text{inf}\left(\mathcal{Y}\right)$. 

\subsection{Identification strategy: time-varying binarization}

\label{subsec:TVB}

It is well-known (see e.g. \citet{DoksumGasko1990} and \citet{Chen2002})
that cross-sectional transformation models can be \emph{binarized}
into a set of binary choice models. Binarization has been used previously
in panel settings (see e.g. \citet{Chen2010}, \citet{ChernozhukovWP2018}),
but those approaches have restricted the threshold to be equal across
time periods. An exception is \citet{Muris}, who uses time-varying
thresholds in a panel ordered logit model with a time-invariant and
parametric transformation.

We now describe \emph{time-varying binarization}. For an arbitrary
threshold $y_{t}\in\underline{\mathcal{Y}}\equiv\mathcal{Y}\backslash\inf\mathcal{Y}$,\footnote{We use $\underline{\mathcal{Y}}$ instead of $\mathcal{Y}$ because
$D_{t}\left(\inf\mathcal{Y}\right)=1$ almost surely for all $t$.} we define the following binary random variable:
\begin{align}
D_{t}\left(y_{t}\right) & \equiv1\left\{ Y_{t}\geq y_{t}\right\} \label{eq:defineDt-1-1}\\
 & =1\left\{ U_{t}\leq\alpha+X_{t}\beta-h_{t}^{-}\left(y_{t}\right)\right\} ,\nonumber 
\end{align}
where the equality follows from specification \eqref{eq:FELT} and
weak monotonicity of $h_{t}$. Varying the threshold $y_{t}$ in (\ref{eq:defineDt-1-1})
across time periods converts any FELT model into a collection of binary
choice models.

Two time periods are sufficient for our identification results, so
we let $T=2$ in what follows. Consider the vector of binary variables
\[
D\left(y_{1},y_{2}\right)\equiv\left(D_{1}\left(y_{1}\right),D_{2}\left(y_{2}\right)\right),
\]
for any two points $(y_{1},y_{2})\in\mathcal{\underline{\mathcal{Y}}}^{2}$.
Our identification strategy for $\left(\beta,h_{1},h_{2}\right)$
is based on the observation that $D\left(y_{1},y_{2}\right)$ follows
a fixed effects binary choice model for \textit{any} $(y_{1},y_{2})\in\mathcal{\underline{\mathcal{Y}}}^{2}$.
This result is summarized in Lemma \ref{lem:Manskimedsign} below.

The proof of identification proceeds in three steps. First, we show
identification of $\beta$ and of $h_{2}^{-}\left(y_{2}\right)-h_{1}^{-}\left(y_{1}\right)$
for arbitrary $\left(y_{1},y_{2}\right)\in\mathcal{\underline{\mathcal{Y}}}^{2}$.
In the resulting binary choice model relating $D_{t}\left(y\right)$
to $X$, the difference $h_{2}^{-}\left(y_{2}\right)-h_{1}^{-}\left(y_{1}\right)$
is the coefficient on the differenced time dummy, and $\beta$ is
the regression coefficient on $X_{2}-X_{1}$. For a given binary choice
model, identification of $\beta$ and of $h_{2}^{-}\left(y_{2}\right)-h_{1}^{-}\left(y_{1}\right)$
follows \citet{Manski1987} for the nonparametric version of FELT,
and \citet{Chamberlain2010} for the logistic version. This result
is summarized in Theorem \ref{thm:ID1} below.

Second, we show that varying the pair $\left(y_{1},y_{2}\right)$
over $\underline{\mathcal{Y}}^{2}$ obtains identification of 
\[
\left\{ h_{2}^{-}\left(y_{2}\right)-h_{1}^{-}\left(y_{1}\right),\,\left(y_{1},y_{2}\right)\in\underline{\mathcal{Y}}^{2}\right\} .
\]

Third, we show that identification of this set of differences obtains
identification of $h_{1}$ and $h_{2}$ under a normalization assumption
on $h_{1}^{-}$.%
\begin{comment}
That is, for an arbitrary $y_{0}\in\underline{\mathcal{Y}}$, $h_{1}^{-}\left(y_{0}\right)=0$.
This type of assumption is customarily made in the literature on transformation
models. 
\end{comment}
{} This result is presented in Theorem \ref{thm:identify-h}.

Figures 3.1 and 3.2 illustrate the intuition behind our identification
strategy for two arbitrary functions, $h_{1}$ and $h_{2}$, both
accommodated by FELT. The line with kinks and a flat part represents
an arbitrary function $h_{1}$, while the solid curve represents an
arbitrary function $h_{2}$. Consider Figure 3.1. Pick a $y_{1}\in\mathcal{\underline{\mathcal{Y}}}$
on the vertical axis. For all $y\leq y_{1}$, $h_{1}\left(y\right)$
gets mapped to zero, while for all $y>y_{1}$, it gets mapped to one.
Now pick a $y_{2}\in\underline{\mathcal{Y}}$. For all $y\leq y_{2}$,
$h_{2}\left(y\right)$ gets mapped to zero, while for all $y>y_{2}$,
it gets mapped to one. This gives rise to a fixed effects binary choice
model for $\left(D_{1}\left(y_{1}\right),D_{2}\left(y_{2}\right)\right)$,
also plotted in the figure as the grey solid lines. Our first result
in Theorem \ref{thm:ID1} identifies the difference $h_{1}^{-}\left(y_{1}\right)-h_{2}^{-}\left(y_{2}\right)$
at arbitrary points $\left(y_{1},y_{2}\right)$, as well as the coefficient
$\beta$. It is clear that normalizing $h_{1}^{-}\left(.\right)$
at an arbitrary point identifies the function $h_{2}^{-}\left(y_{2}\right)$
at an arbitrary point $y_{2}$. This is captured in Figure 3.2. There,
for an arbitrary $y_{0}\text{, }h_{1}^{-}\left(y_{0}\right)=0.$ Then,
as $y_{2}$ is arbitrary, Figure 3.2 shows that moving $y_{2}$ on
its support traces out the generalized inverse $h_{2}^{-}$ on its
domain. Theorem \ref{thm:identify-h} wraps up this argument by showing
that $h_{1}$ and $h_{2}$ are identified from their generalized inverses.%
\begin{comment}
This allows us to group individuals into \textit{switchers} and \textit{non-switchers},
where an individual is a switcher provided that $D_{1}\left(y_{1}\right)+D_{2}\left(y_{2}\right)=1$.
It is the existence of switchers that informs our identification of
the time-varying transformation.
\end{comment}

\begin{figure}
\begin{minipage}[c][1\totalheight][t]{0.45\textwidth}%
\includegraphics[scale=0.6]{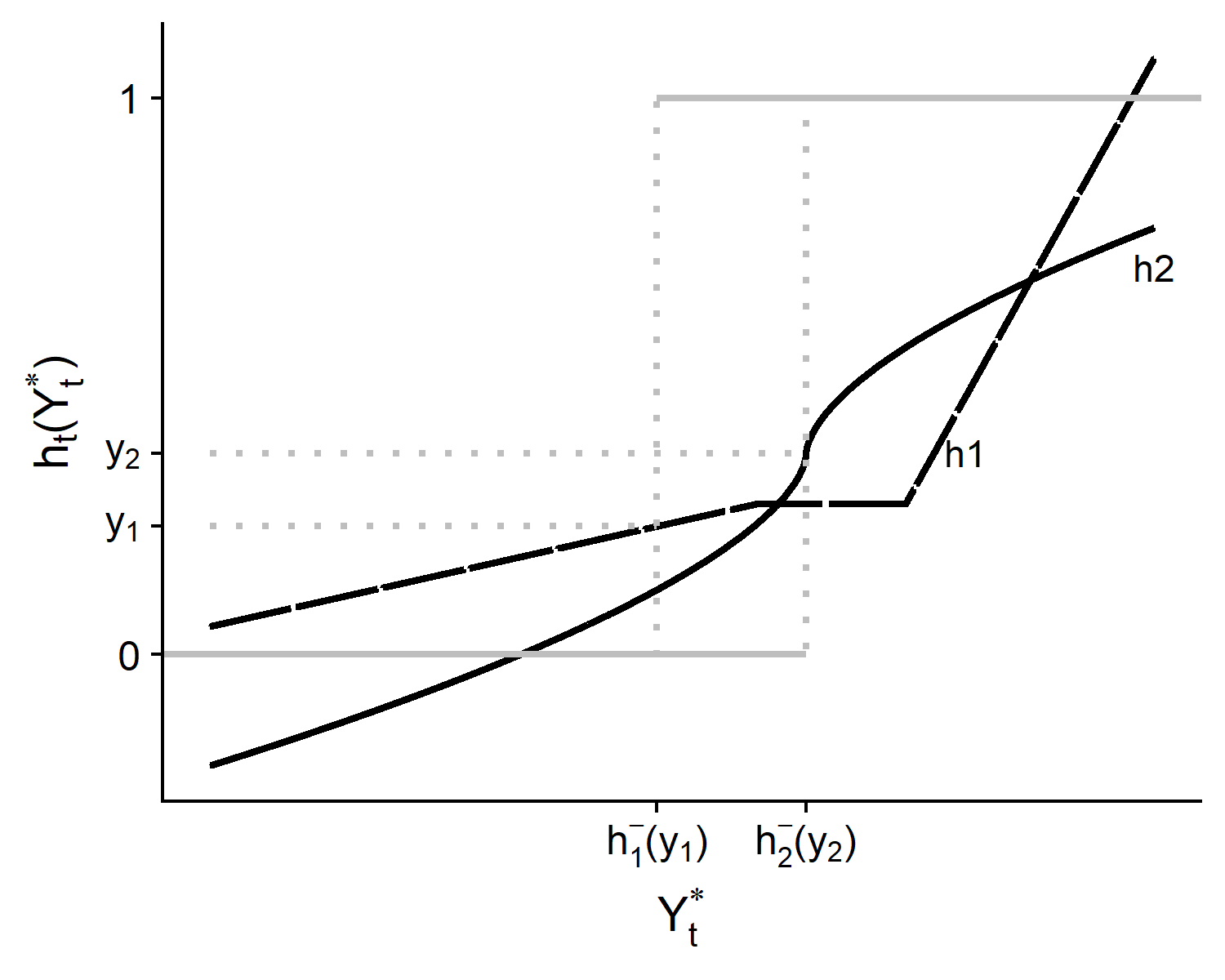}

\caption{FELT functions $h_{1}$ and $h_{2}$.}
\end{minipage}\hfill{}%
\begin{minipage}[c][1\totalheight][t]{0.45\textwidth}%
\includegraphics[scale=0.6]{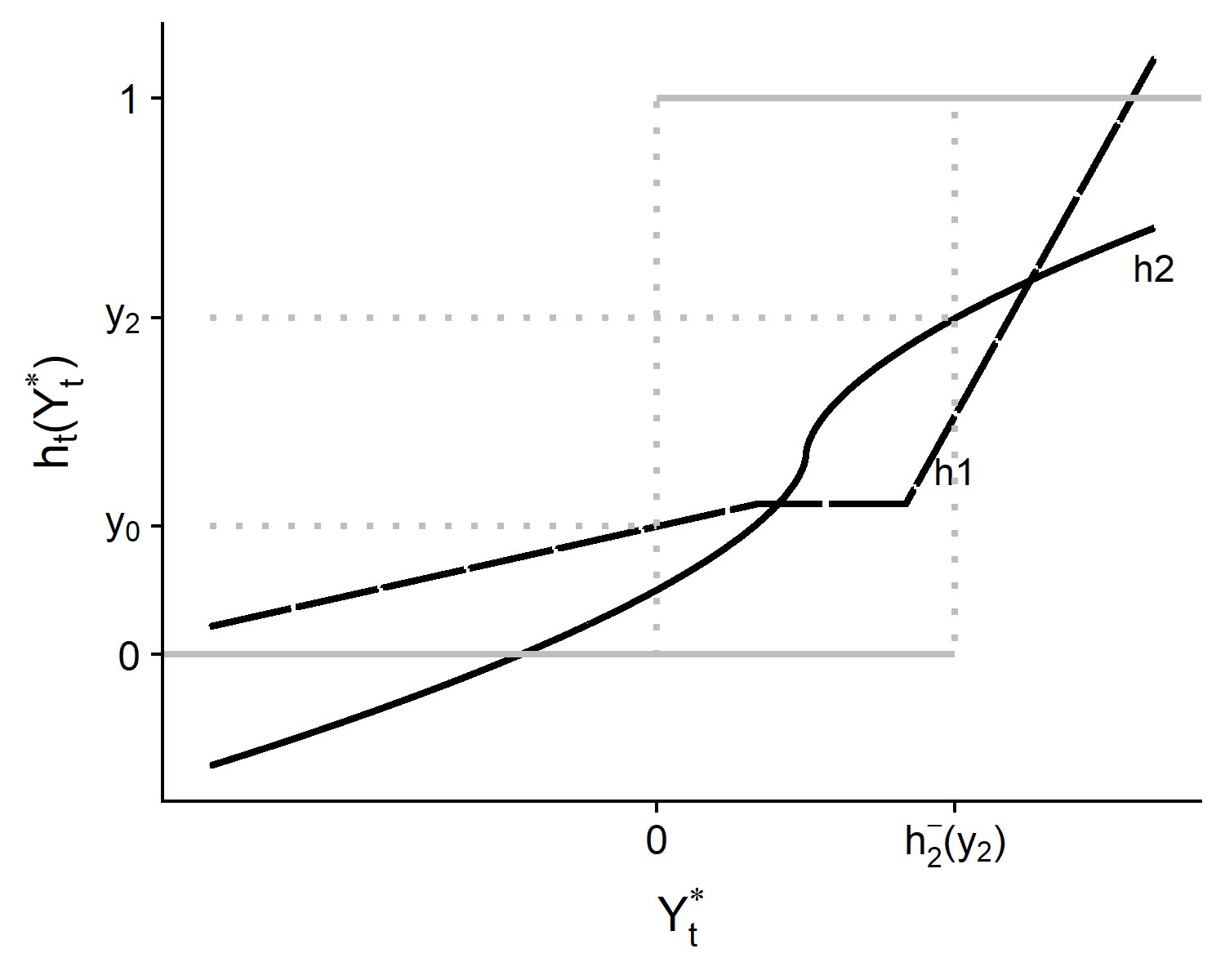}

\caption{Normalization and tracing.}
\end{minipage}
\end{figure}

\subsection{Nonparametric errors\label{subsec:id-Manski}}

In this section, we provide nonparametric identification results for
$\left(\beta,h_{1},h_{2}\right)$. Parts of our identification proof
build on \citet{Manski1987}, who in turn builds on \citet{manski_maximum_1975,Manski1985}.
\begin{assumption}
\label{A2} {[}Error terms{]}

(i) $F_{1}(u|\alpha,X)=F_{2}(u|\alpha,X)\equiv F(u|\alpha,X)$ for
all $(\alpha,X)$;

(ii) The support of $F(\left.u\right|\alpha,X)$ is $\mathbb{R}$
for all $\left(\alpha,X\right)$.
\end{assumption}
Assumption \ref{A2} places no parametric distributional restrictions
on the distribution of $U_{it}$ and allows the stochastic errors
$U_{it}$ to be correlated across time. The first part of the assumption,
\ref{A2}(i), is a stationarity assumption, requiring time-invariance
of the distribution of the error terms conditional on the trajectory
of the observed regressors and on the unobserved heterogeneity. This
assumption excludes lagged dependent variables as covariates. Additionally,
as noted by, e.g., \citet{Chamberlain2010}, although it allows for
heteroskedasticity, it restricts the relationship between the observed
regressors and $U_{it}$ by requiring that even when $x_{1}\not=x_{2}$,
$U_{1}$ and $U_{2}$ have equal skedasticities. This type of stationarity
assumption is common in linear and nonlinear panel models, e.g., \citet{ChernValHahnNewey2013}
and references therein.

Assumption \ref{A2}(ii) requires full support of the error terms.
It guarantees that, for any pair $\left(y_{1},y_{2}\right)\in\underline{\mathcal{Y}}^{2}$,
the probability of being a \textit{switcher }is positive\textit{.}
In our context, being a switcher refers to the event $D_{1}\left(y_{1}\right)+D_{2}\left(y_{2}\right)=1$,
so that Assumption \ref{A2} guarantees that $P\left(D_{1}\left(y_{1}\right)+D_{2}\left(y_{2}\right)=1\right)>0$.
This assumption is similar to Assumption 1 in \citet{Manski1987}.

Let $\Delta X\equiv X_{2}-X_{1}\,$ and for an arbitrary pair $(y_{1},y_{2})\in\mathcal{\underline{Y}}^{2},$
define 
\begin{align}
\gamma\left(y_{1},y_{2}\right) & \equiv h_{2}^{-}\left(y_{2}\right)-h_{1}^{-}\left(y_{1}\right).\label{eq:definegamma}
\end{align}

\begin{lem}
\label{lem:Manskimedsign}Suppose that $\left(Y,X\right)$ follows
the model in \eqref{eq:model}. Let Assumptions \ref{A1} and \ref{A2}
hold. Then for all $\left(y_{1},y_{2}\right)\in\underline{\mathcal{Y}}^{2},$
\begin{equation}
\operatorname{med}\left(D_{2}\left(y_{2}\right)-D_{1}\left(y_{1}\right)|X,\,D_{1}\left(y_{1}\right)+D_{2}\left(y_{2}\right)=1\right)=\operatorname{sgn}\left(\Delta X\beta-\gamma\left(y_{1},y_{2}\right)\right).\label{eq:med_sgn}
\end{equation}
\end{lem}
\begin{proof}
The proof builds on \citet{Manski1985,Manski1987}, and is presented
in Appendix \ref{subsec:Proof-of-Lemma_mediansign}.
\end{proof}
Let $W\equiv(\Delta X,-1)'$ and $\theta\left(y_{1},y_{2}\right)\equiv\left(\beta,\gamma\left(y_{1},y_{2}\right)\right),$
so that $\left(\ref{eq:med_sgn}\right)$ can be written as
\[
\text{med}\left(D_{2}\left(y_{2}\right)-D_{1}\left(y_{1}\right)|X,D_{1}\left(y_{1}\right)+D_{2}\left(y_{2}\right)=1\right)=\text{sgn}\left(W\theta\left(y_{1},y_{2}\right)\right).
\]
For identification of $\theta\left(y_{1},y_{2}\right)$ we impose
the following additional assumptions.
\begin{assumption}
\label{A3}{[}Covariates{]}

(i) The distribution of $\Delta X$ is such that at least one component
of $\Delta X$ has positive Lebesgue density on $\mathbb{R}$ conditional
on all the other components of $\Delta X$ with probability one. The
corresponding component of $\beta$ is non-zero.

(ii) The support of $W$ is not contained in any proper linear subspace
of $\mathbb{R}^{K+1}$.
\end{assumption}
Assumption \ref{A3}(i) requires that the change in one of the regressors
be continuously distributed conditional on the other components. Assumption
\ref{A3}(ii) is a full rank assumption. These assumptions are standard
in the binary choice literature concerned with point identification
of the parameters.

Assumption \ref{A3} resembles Assumption 2 in \citet{Manski1987},
the difference being that our assumption concerns $W$, which includes
a constant that captures a time trend. The presence of this constant
requires sufficient variation in $X_{t}$ over time. No linear combination
of the components of $X_{t}$ can equal the time trend.
\begin{assumption}
\label{ass:normalization}{[}Normalization-$\beta${]} For any $(y_{1},y_{2})\in\mathcal{\underline{Y}}^{2},$
$\theta\left(y_{1},y_{2}\right)\in\Theta=\mathcal{B}\times\mathbb{R},$
where $\mathcal{B}=\left\{ \beta:\beta\in\mathbb{R}^{K},\|\beta\|=1\right\} $.
\end{assumption}
Assumption \ref{ass:normalization} imposes a normalization on $\beta$,
namely that the norm of the regression coefficient equals 1. Scale
normalizations are standard in the binary choice literature, and are
necessary for point identification when the distribution of the error
terms is not parameterized. Normalizing $\beta$ (instead of $\theta$)
avoids a normalization that would otherwise depend on the choice of
$\left(y_{1},y_{2}\right)$. In this way, the scale of $\beta$ remains
constant across different choices of $\left(y_{1},y_{2}\right)$.
Alternatively, one can normalize the coefficient on the continuous
covariate (cf. Assumption \ref{A3}(i)) to be equal to one. In our
economic model in Section \ref{sec:MicroModel} the latter assumption
holds automatically.\footnote{There are models with sufficient structure on the transformation $h_{t}$
where identification is possible without a normalization on the regression
coefficient. Examples include the linear regression model, the censored
linear regression model in \citet{Honore1992}, and the interval-censored
regression model in \citet{abrevaya_interval_2020}.}
\begin{thm}
\label{thm:ID1}Suppose that $\left(Y,X\right)$ follows the model
in \eqref{eq:model}, and let the distribution of $\left(Y,X\right)$
be observed. Let Assumptions \ref{A1}, \ref{A2}, \ref{A3}, and
\ref{ass:normalization} hold. Then, for an arbitrary pair $\left(y_{1},y_{2}\right)\in\underline{\mathcal{Y}}^{2},$
$\theta\left(y_{1},y_{2}\right)$ is identified.
\end{thm}
\begin{proof}
The proof proceeds by showing that FELT can be converted into a binary
choice model for an arbitrary pair $\left(y_{1},y_{2}\right),$ and
then builds on Theorem 1 in \citet{Manski1987}, which in turn uses
results in \citet{Manski1985}. See Appendix \ref{subsec:Proof-of-Theorem_manski_ID}.
\end{proof}
So far, we have identified the regression coefficient $\beta$ and
the difference in the generalized inverses at arbitrary pairs $\left(y_{1},\,y_{2}\right)$.
We consider now identification of the functions $h_{1}$ and $h_{2}$
on $\underline{\mathcal{Y}}$.
\begin{assumption}
\label{ass:normalizeh}{[}Normalization-$h_{1}${]} For some $y_{0}\in\mathcal{\underline{\mathcal{Y}}},$
$h_{1}^{-}\left(y_{0}\right)=0$.
\end{assumption}
Such a normalization is standard in transformation models, see, e.g.,
\citet{Horowitz1996}. Without this normalization, all identification
results hold up to $h_{1}^{-}\left(y_{0}\right)$. We normalize the
function in the first time period only, imposing no restrictions on
the function in the second period beyond that of weak monotonicity
(cf. Assumption \ref{A1}). In Section \ref{sec:FE}, we show that
this normalization assumption is not necessary for our results on
the identification of the conditional mean or of the conditional variance
of the fixed effects.
\begin{thm}
\label{thm:identify-h}Suppose that $\left(Y,X\right)$ follows the
model in \eqref{eq:model}, and let the distribution of $\left(Y,X\right)$
be observed. Under Assumptions \ref{A1}, \ref{A2}, \ref{A3}, \ref{ass:normalization},
and \ref{ass:normalizeh}, the transformations $h_{1}$ and $h_{2}$
are identified.
\end{thm}
\begin{proof}
The proof proceeds by identifying the generalized inverses of monotone
functions, which obtains identification of the pre-images of $h_{1}$
and $h_{2}$. This obtains identification of the functions themselves.
See Appendix \ref{subsec:Proof-of-Theorem_identificationfunctions}.
\end{proof}

\subsection{Logit errors\label{subsec:Logistic-FELT_model}}

In this section, we show identification of $\left(\beta,h_{1},h_{2}\right)$
when the error terms are assumed to follow the standard logistic distribution.
The logistic case is not nested in the nonparametric case. In particular,
when the errors are logistic, we do not require a continuous regressor.
However, we require conditional serial independence of the error terms.\footnote{See \citet{Chamberlain2010} and \citet{Magnac2004} for more details
about identification under nonparametric versus logistic errors in
the panel data binary choice context.}
\begin{assumption}
\label{ass:logit}{[}Logit{]} (i) $F_{1}(u|\alpha,X)=F_{2}$$\left(u|\alpha,X\right)=\Lambda\left(u\right)=\frac{\exp\left(u\right)}{1+\exp\left(u\right)}$,
and $U_{1}$ and $U_{2}$ are independent; (ii) $E(W'W)$ is invertible.
\end{assumption}
Assumption \ref{ass:logit}(i) strengthens Assumption \ref{A2} by
requiring the errors to follow the standard logistic distribution
and to be serially independent. Note that one consequence of this
assumption, which specifies the variance of the error terms to be
equal to 1, is to eliminate the need to normalize $\beta$. On the
other hand, Assumption \ref{ass:logit}(ii) imposes weaker restrictions
on the observed covariates relative to Assumption \ref{A3}, since
it does not require the existence of a continuous covariate. Sufficient
variation in $\Delta X$ is sufficient to obtain identification of
the vector $\beta$ when the error terms follow the standard logistic
distribution.
\begin{thm}
\label{thm:logitID}Suppose that $\left(Y,X\right)$ follow the model
in \eqref{eq:model}, and let the distribution of $\left(Y,X\right)$
be observed. Let Assumptions \ref{A1} and \ref{ass:logit} hold.
Then, for an arbitrary pair $\left(y_{1},y_{2}\right)\in\mathcal{\underline{Y}}^{2},$
$\theta\left(y_{1},y_{2}\right)$ is identified. Additionally, letting
Assumption \ref{ass:normalizeh} hold, then the transformations $h_{1}\left(\cdot\right)$
and $h_{2}\left(\cdot\right)$ are identified.
\end{thm}
\begin{proof}
See Appendix \ref{subsec:Proof-of-Theorem_identificationlogit}.
\end{proof}

\section{Conditional distribution of fixed effects\label{sec:FE}}

If $(h_{1},h_{2})$ are invertible, we can use the previous identification
theorem to identify features of the distribution of the fixed effects
conditional on observed regressors. These features are the change
in the conditional mean function of $\alpha$ and the conditional
variance of $\alpha$ conditional on $X_{1},X_{2}$. These results
are relevant since in our collective household model, the fixed effects
represent the log of resource shares, and both the standard deviation
of these resource shares and the response of their conditional mean
to covariates are key parameters of interest in the empirical literature.
As this is relevant to our application, we note here that a normalization
assumption, such as \ref{ass:normalization}, on the demand function
in the first period is not necessary for these results on the resource
shares because, e.g., we only need their deviation with respect to
the mean of the fixed effects.

In this section, we provide sufficient conditions for the identification
of the change in the conditional mean function of the fixed effects,
defined as:
\begin{equation}
\mu\left(x\right)\equiv E\left[\left.\alpha\right|X=x\right],\text{ for all \ensuremath{x\in\mathcal{X},}}
\end{equation}
as well as for the conditional variance of the fixed effects. For
these results, the normalization assumption \ref{ass:normalizeh}
is not necessary. To provide intuition for this, let 
\[
c_{1}\equiv h_{1}^{-1}\left(y_{0}\right),
\]
at an arbitrary $y_{0}\in\mathcal{\underline{Y}}$ and $g_{t}\left(y\right)\equiv h_{t}^{-1}\left(y\right)-c_{1}$
for all $y\in\mathcal{\underline{Y}}$. Note that Theorem \ref{thm:ID1}
recovers 
\[
\widetilde{U}_{t}\equiv\alpha-U_{t}=h_{t}^{-1}\left(Y_{t}\right)-X_{t}\beta,
\]
up to $c_{1}$, so that the joint distribution of $\left(\widetilde{U}_{1},\widetilde{U}_{2},X_{1},X_{2}\right)$
is identified up to $c_{1}$. By placing restrictions on the distribution
of $\left(\alpha,U_{1},U_{2},X_{1},X_{2}\right)$, we can then recover
our features of interest.
\begin{thm}
\label{thm:conditional-mean-a-ID}(i) Let the assumptions of Theorem
\ref{thm:ID1} hold, and additionally assume that (4a) $\left(h_{1},h_{2}\right)$
are strictly\textbf{ }increasing, and (4b) let $m\in\mathbb{R}$ be
an unknown constant such that $E\left(\left.U_{t}\right|X=x\right)=m$,
for all $x\in\mathcal{X}$. Then, for any $x,x'\in\mathcal{X}$, the
change in the conditional mean function $\mu\left(x\right)-\mu\left(x'\right)$
is identified and given by
\[
\mu\left(x\right)-\mu\left(x^{'}\right)=E\left[\left.g_{t}\left(Y_{t}\right)-X_{t}\beta\right|X=x\right]-E\left[\left.g_{t}\left(Y_{t}\right)-X_{t}\beta\right|X=x^{'}\right].
\]
\end{thm}
\begin{proof}
See Appendix \ref{subsec:Proof-of-Theorem_conditionalmean}.
\end{proof}
\begin{rem}
As opposed to our main identification result in Theorem \ref{thm:identify-h},
Theorem \ref{thm:conditional-mean-a-ID} does not use a normalization
on the functions $\left(h_{1},h_{2}\right)$. If we were to impose
the normalization in Assumption \ref{ass:normalizeh}, the conditional
mean function $\mu\left(x\right)$ would be identified for all $x\in\mathcal{X}$.
This result provides justification for nonparametric regression of
$\alpha$ on observables (up to location).
\end{rem}
\begin{rem}
Under slightly weaker conditions, we can obtain the projection coefficients
of $\alpha$ on $X_{t}$. This is of interest for our empirical application.
Recall that the joint distribution of $\left(\widetilde{U}_{1},\widetilde{U}_{2},X_{1},X_{2}\right)$
is identified up to $c_{1}$. Then, assuming $\text{Cov}\left(U_{s},X_{t}\right)=0$,
we can identify the projection coefficient of $\alpha$ on $X_{t}$
from 
\[
\left[\text{Var}\left(X_{t}\right)\right]^{-1}\text{Cov}\left(\alpha,X_{t}\right)=\left[\text{Var}\left(X_{t}\right)\right]^{-1}\text{Cov}\left(\alpha-U_{s},X_{t}\right)=\left[\text{Var}\left(X_{t}\right)\right]^{-1}\text{Cov}\left(\widetilde{U}_{s},X_{t}\right).
\]
\end{rem}
Second, define the conditional variance of the fixed effects as
\begin{equation}
\sigma_{\alpha}^{2}\left(x\right)\equiv Var\left[\left.\alpha\right|X=x\right],\text{ for all \ensuremath{x\in\mathcal{X}.}}
\end{equation}
For this second result, we strengthen our assumptions to include,
among others, serial independence of the error term. This allows us
to pin the persistence in unit $i$'s time series on $\alpha_{i}$
instead of on serial dependence in the errors.
\begin{thm}
\label{thm:conditional-var-a-ID} Let the assumptions of Theorem \ref{thm:ID1}
hold and assume that (5a) $\left(h_{1},h_{2}\right)$ are strictly\textbf{
}increasing, and (5b) $Cov\left[\left.\alpha,U_{t}\right|X=x\right]=0$
for all $x\in\mathcal{X}$ and $t$, and (5c) $Cov\left[\left.U_{1},U_{2}\right|X=x\right]=0$
for all $x\in\mathcal{X}$. Then for all $x\in\mathcal{X},$ the conditional
variance function $\sigma_{\alpha}^{2}\left(x\right)$ is identified
and given by:
\[
\sigma_{\alpha}^{2}\left(x\right)=Cov\left(\left.g_{2}\left(Y_{2}\right)-X_{2}\beta,g_{1}\left(Y_{1}\right)-X_{1}\beta\right|X=x\right).
\]
\end{thm}
\begin{proof}
See Appendix \ref{subsec:Proof-of-Theorem_conditionalmean}.
\end{proof}
It may be possible to obtain the entire conditional distribution of
the fixed effects under the assumption that $\left(\alpha,U_{1},U_{2}\right)$
are mutually independent by using arguments similar to those in \citet{AB_distributional_2012}.

\section{Microeconomic model\label{sec:MicroModel}}

In this section, we construct a model of an efficient full-commitment
intertemporal collective (FIC) household. We provide sufficient conditions
implying that observed household-level quantity demand equations are
time-varying functions of fixed effects that lie in the FELT class.

Essentially, %
\begin{comment}
we combine the models of \citet{BrowningChiapporiLewbel13} and \citet{ChiapporiMazzocco2017}
to generate an empirically practical model that allows identification
of resource shares.
\end{comment}
\begin{comment}
\citet{ChiapporiMazzocco2017} write their model in terms of pure
public and pure private goods.
\end{comment}
we %
\begin{comment}
instead
\end{comment}
{} adapt the intertemporal collective household model of \citet{ChiapporiMazzocco2017}
to the model of household scale economies given in the static household
model of \citet{BrowningChiapporiLewbel13}. We assume efficiency,
i.e. that the household members together reach the Pareto frontier.
Consequently, our model does not account for inefficiency due to e.g.
consumption externalities or information asymmetries. %
\begin{comment}
A feature of efficient models like this is that the household-level
problem can be decentralized into an observationally equivalent set
of individual decision problems. Each individual problem is to choose
demands based on an individual-level constraint defined by a shadow
price vector and a shadow budget constraint.
\end{comment}

We use subscripts $i,j,t$. Let $i=1,...,n$ index households and
assume the household has a time-invariant composition, with $N_{ij}$
members of type $j$. Let $j=m,f,c$ for men, women and children.
Let $t=1,2$. Let $z$ be a vector of time-varying household-level
demographic characteristics, and let the numbers of household members
of each type, $N_{im}$, $N_{if}$ and $N_{ic}$, be (time-invariant)
elements of $z_{it}$. Like \citet{ChiapporiMazzocco2017}, this is
a model with uncertainty, so we use the superscript $s=1,2$ to index
states in the second period only. As in \citet{ChiapporiMazzocco2017},
the use of 2 time periods and 2 states is for illustration only; the
model goes through with any finite number of states or periods. Similarly,
inclusion of a risky asset would not change the features of the model
that we use.

Indirect utility, $V_{j}(p,x,z)$, is the maximized value of utility
given a budget constraint defined by prices $p$ and budget $x$,
given characteristics $z$. Let $V_{j}$ be strictly concave in the
budget $x$. Indirect utility depends on time only through its dependence
on the budget constraint and time-varying demographics $z$. Let $v_{ijt}\equiv V_{j}(p_{t},x_{it},z_{it})$
denote the utility level of a person of type $j$ in household $i$
in period $t$. 

\begin{comment}
\citet{BrowningChiapporiLewbel13} model sharing and household scale
economies via a household consumption function that reflects the fact
that shareable goods feel ``cheap'' within the household. This is
embodied in a shadow price vector for consumption within the household
that is weakly smaller than the market price vector $p_{t}$ faced
by single individuals, because singles cannot take advantage of scale
economies in household consumption. For example, goods that are not
shareable at all---for which there are no scale economies in household
consumption---have shadow prices equal to the market price. Goods
that are fully shareable, so that each person in the household can
enjoy an effective consumption equal to the amount purchased by the
household, have a shadow price equal to the market price divided by
the number of members. 
\end{comment}

Household decisions are made on the basis of individual shadow budget
constraints, reflecting the economic environment within the household.
These shadow constraints are characterized by a shadow price vector
faced by all members and shadow budgets which may differ across household
members. Let $A_{it}\equiv A(z_{it})$ be a diagonal matrix that gives
the shareability of each good, and let it depend on demographics $z_{it}$
(including the numbers of household members). More shareable goods
have lower shadow prices of within-household consumption.

For nonshareable goods, the corresponding element of $A_{it}$ equals
$1$; for shareable goods, it is less than $1$, possibly as small
as $1/N_{i}$ where $N_{i}$ is the number of household members. Goods
may be partly shareable, with an element of $A_{it}$ between $1/N_{i}$
and $1$. With market prices $p_{t}$, within-household shadow prices
are given by the linear transformation $A_{it}p_{t}$. %
\begin{comment}
Shadow prices are the same for all household members $j$. 
\end{comment}

\citet{BrowningChiapporiLewbel13} also allow for inequality in shadow
budgets. Let $\eta_{ijt}$ be the \emph{resource share} of type $j$
in household $i$ in time period $t$. It gives the fraction of the
household budget consumed by that type. Each person of the $N_{ij}$
people of type $j$ consumes $\eta_{ijt}/N_{ij}$ of the household
budget $x_{it}$, so they each have a shadow budget of $\eta_{ijt}x_{it}/N_{ij}$.%
\begin{comment}
As we will see below, the resource share is a choice variable for
the household.
\end{comment}

\begin{comment}
The resource shares and shadow price vector together define the decentralized
shadow budget constraints faced by each household member. 
\end{comment}
Each household member faces a shadow budget of $\eta_{ijt}x_{it}/N_{ij}$
and shadow prices of $A_{it}p$$_{t}$, so that, within the household,
indirect utility is given by 
\begin{equation}
v_{ijt}=V_{j}\left(A_{it}p_{t},\eta_{ijt}x_{it}/N_{ij},z_{it}\right).\label{eq:within hh V}
\end{equation}
Let $V_{xj}(p,w,z)\equiv\partial V_{j}(p,w,z)/\partial w$ be the
monotonically decreasing marginal utility of person $j$ with respect
to their shadow budget. Then, $V_{xj}\left(A_{it}p_{t},\eta_{ijt}x_{it}/N_{ij},z_{it}\right)$
is the value of their marginal utility evaluated at their shadow budget
constraint.

Let $p_{2}^{s},z_{i2}^{s},\overline{x}_{i}^{s}$ for $s=1,2$ be the
possible realizations of state-dependent variables that occur with
household-specific probabilities $\pi_{i}^{1}$ and $\pi_{i}^{2}$
(which sum to $1$). Here, $\overline{x}_{i}^{s}$ is the state-specific
lifetime wealth of household $i$, revealed in period $2$. 

Pareto weights $\phi_{ij}$ vary arbitrarily across households, and
depend, for example, on the household-specific expectation of lifetime
wealth at the moment of household formation. Because ours is a full
commitment model, Pareto weights do not have a time subscript because
they are fixed at the moment of household formation, when the full-commitment
contract is set. This household-level time-invariant variable will
form the basis of our fixed-effects variation.%
\begin{comment}
Information about the joint distribution of these unobserved state-dependent
variables is embodied in $\Phi_{i}$, the information set available
in period $1$ to household $i$.
\end{comment}
{} %
\begin{comment}
The information set can also include unobserved time-invariant constraints
faced by the household, for example, unobserved variation in expected
household wealth. 
\end{comment}

As in \citet{ChiapporiMazzocco2017},%
\begin{comment}
let individuals have expected lifetime utilities given by the sum
of period $1$ utility and the discounted probability-weighted sum
of state-specific period $2$ utility.\footnote{If the membership of the household changed over time, or if the household
was choosing its membership, we would need a household welfare function
that used some kind of population ethics principle (see \citet{Blackorby2005}).
It is for this reason that we focus on households with fixed membership.

Here, we only consider egotistic preferences. However, this is can
easily be generalized to a model with caring preferences: ``It is
important to point out, however, that the model with egotistical preferences
... plays a special role. The reason for this is that the solution
to the collective model with caring preferences must also be a solution
of the collective model ... with egotistical preferences.'' (\citet{ChiapporiMazzocco2017},
page 21).

The restriction there there are only two periods and only two states
is for convenience. None of our conclusions about resource shares
depend on it.} they then write 
\end{comment}
the assumption of efficiency implies that we can represent the household's
decisions via the Bergson-Samuelson Welfare Function, $W_{i}$, for
the household is

\begin{equation}
W_{i}\equiv\sum_{j=1}^{J}N_{ij}\phi_{ij}\left[v_{ij1}+\sum_{s=1}^{2}\rho_{i}\pi_{i}^{s}v_{ij2}^{s}\right].\label{eq: welfare}
\end{equation}
The term in square brackets is the expected lifetime utility, discounted
by the discount factor $\rho_{i}$, of each member of type $j$ in
household $i$. Each member of type $j$ gets the Pareto weight $\phi_{ij}$.
The assumption of efficiency implies that the household reaches the
Pareto Frontier; the Pareto weights pinpoint which point on the Frontier
is chosen by the household. %
\begin{comment}
These Pareto weights are fixed at the moment of household formation,
and therefore time-invariant.
\end{comment}

Next, substitute indirect utility (\ref{eq:within hh V}) for utility
$v_{ij1}$ and $v_{ij2}^{s}$ into (\ref{eq: welfare}),\footnote{In contrast, \citet{ChiapporiMazzocco2017} substitute \emph{direct}
utility for utility $v_{ij1}$ and $v_{ij2}^{s}$ using a model of
pure private and pure public goods. In that model, each individual's
utility is given by their direct utility function, which is a function
of their (unobserved) consumption of a vector of private goods and
their (observed) consumption of a vector of public goods. } and form the Lagrangian using the intertemporal budget constraint
with interest rate $\tau$, $x_{i1}+x_{i2}^{s}/(1+\tau)=\overline{x}_{i}^{s}$,
and the adding-up constraints on resource shares, $\sum_{j}\eta_{ij1}=\sum_{j}\eta_{ij2}^{s}=1.$
Each household $i$ chooses $x_{i1},$ $\eta_{ij1}$ and $\eta_{ij2}^{s}$
to maximize $W_{i}$.%
\begin{comment}
\[
W_{i}=\sum_{j=1}^{J}N_{ij}\phi_{ij}\left[V_{j}(A_{i1}p_{1},\eta_{ij1}x_{i1}/N_{ij},z_{i1})+\sum_{s=1}^{2}\rho_{i}\pi_{i}^{s}V_{j}(A_{i2}p_{2}^{s},\eta_{ij2}^{s}x_{i2}^{s}/N_{ij},z_{i2}^{s})\right]
\]
\[
-\sum_{s=1}^{2}\kappa^{s}\left(x_{i1}+x_{i2}^{s}/(1+\tau)-\overline{x}_{i}^{s}\right)-\lambda_{1}\left(\sum_{j=1}^{J}\eta_{ij1}-1\right)-\sum_{s=1}^{2}\lambda_{2}^{s}\left(\sum_{j=1}^{J}\eta_{ij2}^{s}-1\right).
\]

First-order conditions for $\eta_{ijt}$ in each period and each state
are given by:
\begin{align*}
\phi_{ij}V_{xj}(A_{i1}p_{1},\eta_{ij1}x_{i1}/N_{ij},z_{i1})x_{i1}-\lambda_{1} & =0,\\
\rho_{i}\phi_{ij}V_{xj}(A_{i2}p_{2}^{s},\eta_{ij2}^{s}x_{i2}^{s}/N_{ij},z_{i2}^{s})x_{i2}^{s}-\lambda_{2}^{s} & =0,
\end{align*}
for $s=1,2$. 
\end{comment}
\begin{comment}
Because the optimand is additively separable across periods, it can
be thought of as a two-stage budgeting problem, where the household
first chooses the period budgets, $x_{i1}$ and $x_{i2}^{s}$, and
then chooses resource shares conditional on this allocation of budget.
Two-stage budgeting is often used to motivate cross-sectional estimation.
In our context, we will still need panel data to deal with unobserved
household-level heterogeneity (analogous to differencing out such
heterogeneity in linear panels).
\end{comment}
{} Solving this optimization problem obtains for any two types, $j$
and $k$, that:
\begin{equation}
\frac{\phi_{ij}}{\phi_{ik}}=\frac{V_{xk}(A_{i1}p_{1},\eta_{ik1}x_{i1}/N_{ik},z_{i1})}{V_{xj}(A_{i1}p_{1},\eta_{ij1}x_{i1}/N_{ij},z_{i1})}=\frac{V_{xk}(A_{i2}p_{2}^{s},\eta_{ik2}^{s}x_{i2}^{s}/N_{ik},z_{i2}^{s})}{V_{xj}(A_{i2}p_{2}^{s},\eta_{ij2}^{s}x_{i2}^{s}/N_{ij},z_{i2}^{s})},\label{eq: FIC optimal eta}
\end{equation}
for $s=1,2$ for all $i$. That is, the household chooses resource
shares so as to equate ratios of marginal utilities with ratios of
Pareto weights.%
\begin{comment}
There is a unique solution to this problem because each person $j$
has a utility function strictly concave in the shadow budget.
\end{comment}

\begin{comment}
Resource shares are implicitly determined by (\ref{eq: FIC optimal eta}),
are generally time- and state-dependent, and depend on the Pareto-weights
$\phi_{i1},...,\phi_{iJ}$. Because we are in a full-commitment world,
these Pareto-weights are time-invariant. And, because there are both
observed and unobserved household-level shifters to Pareto-weights,
the Pareto-weights are heterogeneous across observably identical households.
Consequently, the Pareto-weights are fixed effects hiding inside the
resource share functions.

Household quantity demands given the sharing model of \citet{BrowningChiapporiLewbel13}
are very simple: the household purchases the sum of what all the individuals
would demand if they faced the within-household shadow price vector
$A_{it}p_{t}$ and had their shadow budget $\eta_{ijt}x/N_{ij}$,
adjusted for sharing as defined by $A_{it}$. A key feature here is
that the household demand for a non-shareable good does not have to
be adjusted for shareability: it is just the sum of what each individual
would demand. 
\end{comment}

An \emph{assignable good} is one where we observe the consumption
of that good by a specific person (or type of person). Assuming the
existence of a scalar-value demand function $q_{j}(p,x,z)$ for an
assignable and non-shareable good (e.g., food or clothing) for a person
of type $j$, \citet{BrowningChiapporiLewbel13} show that the household's
quantity demand, $Q_{ijt}$, for the assignable good for each of the
$N_{ij}$ people of type $j$ is given by
\[
Q_{ijt}=q_{j}(A_{it}p_{t},\eta_{ijt}x_{it}/N_{ij},z_{it}).
\]
\begin{comment}
Because only people of type $j$ purchase this good, the household
does not sum over the demand of other household members. 
\end{comment}
Assuming that the assignable good is a normal good implies that $q_{j}$
is strictly increasing in its second argument, and is therefore strictly
monotonic. 

Since $p_{t}$ is unobserved and varies over time,\footnote{The assumption that prices are unobserved but vary \emph{only} with
time is important for our empirical application below, where we observe
Bangladeshi households in 2 time periods. It implies moment conditions
for differenced demand functions that do not depend on prices. If
prices vary with any observed variables, then these could be worked
into the moment conditions, by conditioning on those variables. However,
if prices vary with unobserved variables, our estimation strategy
would not work.} we can express $Q_{ijt}$ as a time-varying function of observed
data. Defining $\widetilde{q}_{jt}(x,z)=q_{j}(p_{t},x,z)$, we have
\begin{equation}
Q_{ijt}=\widetilde{q}_{jt}(\eta_{ijt}x_{it}/N_{ij},z_{it}).\label{eq: BCL assignable demand}
\end{equation}
This is the structural demand equation that we ultimately bring to
the data.

\begin{comment}
This model is very general. It assumes only that: the household satisfies
the intertemporal budget constraint under uncertainty, can fully commit
to future actions, reaches the Pareto frontier, and has scale economies
embodied in the shareability matrix $A(z)$. It places no additional
restrictions on utility functions or the bargaining model. It implies
that quantity demands for assignable goods are time-varying functions
of resource shares, and that resource shares depend on Pareto weights
that are fixed over time (aka: fixed effects).

\subsection{PIGL Resource Shares}
\end{comment}

The model above determines resource share via the first-order conditions
(\ref{eq: FIC optimal eta}), and these resource shares depend on
time-invariant Pareto weights $\phi_{ij}$. But these resource shares
are a vector of implicit functions, which may be hard to work with.
To make the model tractable, we impose sufficient structure on utility
functions to find closed forms for resource shares. 

In our empirical example below, we work with data that have time-invariant
demographic characteristics, so let $z_{it}=z_{i}$ be fixed over
time. This implies that the shareability of goods embodied in $A$
is time-invariant: $A_{it}=A_{i}=A(z_{i})$. We will estimate the
demand equation for women's food in nuclear households comprised of
$1$ man, $1$ woman and $1-4$ children, so $N_{if}=N_{im}=1$. 

Let indirect utilities be in the price-independent generalized logarithmic
(PIGL) class (\citet{Muellbauer75,Muellbauer76}) given by
\begin{equation}
V_{j}(p,x,z)=C_{j}(p,z)+\left(B(p,z)x\right)^{r(z)}/r(z).\label{eq:simplest PIGL}
\end{equation}
Here, $V_{j}$ is homogeneous of degree $1$ in $p,x$ if $C_{j}$
is homogeneous of degree $0$ in $p$ and $B$ is homogeneous of degree
$-1$ in $p$. $V$ is increasing in $x$ if $B(p,z)$ is positive
and $V$ is concave in $x$ if $r(z)<1$. In terms of preferences,
this class is reasonably wide. It gives quasihomothetic preferences
if $r(z)=1$, and PIGLOG preferences as $r(z)\rightarrow0$ (this
includes the Almost Ideal Demand System of \citet{Deatonm80}). 

The functions $C_{j}$ vary across types $j$, and so the model allows
for preference heterogeneity between types, e.g., between men and
women. The restrictions that $B(p,z)$ and $r(z)$ don't vary across
$j$ and that $r(z)$ does not depend on prices $p$ are important:
as we see below, they imply that resource shares are constant over
time.

Substituting into the BCL model, observed demographics and period
$t$ budgets, we have that %
\begin{comment}
the utility of person $j$ in household $i$ in period $t$ as
\[
v_{ijt}=V_{j}(A_{i}p_{t},\eta_{ijt}x_{it}/N_{ij},z_{i})=C_{j}(A_{i}p_{t},z_{i})+B(A_{i}p_{t},z_{i})^{r(z_{i})}\left(\eta_{ijt}x_{it}/N_{ij}\right)^{r(z_{i})}/r(z_{i}),
\]
and thus
\end{comment}
marginal utilities are given by
\[
V_{xj}(A_{it}p_{t},\eta_{ijt}x_{it}/N_{ij},z_{it})=B(A_{i}p_{t},z_{i})^{r(z_{i})}\left(\eta_{ijt}x_{it}/N_{ij}\right)^{r(z_{i})-1}.
\]
For $r(z_{i})\neq1$, and for any pair of types $j,k\text{, }$we
substitute into (\ref{eq: FIC optimal eta}) and %
\begin{comment}
cancel terms, 
\[
\frac{\phi_{ij}}{\phi_{ik}}=\frac{V_{xk}(A_{it}p_{t},\eta_{ijt}x_{it}/N_{ij},z_{it})}{V_{xj}(A_{it}p_{t},\eta_{ijt}x_{it}/N_{ij},z_{it})}=\left(\frac{\eta_{ikt}/N_{ik}}{\eta_{ijt}/N_{ij}}\right)^{r(z_{i})-1},
\]
Rearranging, we get
\end{comment}
rearrange to get
\begin{equation}
\left(\frac{\eta_{ikt}}{\eta_{ijt}}\right)=\frac{N_{ik}}{N_{ij}}\left(\frac{\phi_{ij}}{\phi_{ik}}\right)^{1/\left(r(z_{i})-1\right)}.\label{eq: resource share invariance}
\end{equation}

The household chooses resource shares in each period and each state
to satisfy (\ref{eq: resource share invariance}). Since the right-hand
side has no variation over time or state, this implies that, given
PIGL utilities (\ref{eq:simplest PIGL}), the resource shares in a
given household $i$ are independent of period $t$ and state $s$.
However, resource shares do vary with both observed and unobserved
variables across households $i$. 

Let the fixed resource shares that solve the first-order conditions
with PIGL demands be denoted $\overline{\eta}_{ij}$, and define 
\begin{equation}
\alpha_{i}\equiv\ln\left(\overline{\eta}_{if}/N_{if}\right)=\ln\overline{\eta}_{if},\label{eq: alpha}
\end{equation}
to equal to the logged resource share of the woman in the household.

Let there be a multiplicative \citet{berkson50} measurement error
denoted $\exp\left(-U\right)$ which multiplies the budget, so that
if we observe $x$, the actual budget is $x/\exp\left(U\right)$.
The measurement error is i.i.d. across time and households, which
implies stationarity of $U$. Here, the measurement error does not
affect resource shares, but does affect the distribution of observed
quantity demands. %
\begin{comment}
Plugging this measurement error and the PIGL form for resource shares
given by (\ref{eq: alpha}) into the assignable goods demand equation
(\ref{eq: BCL assignable demand}) yields
\end{comment}
{} Then household demand for women's food (the assignable good), $Q_{ift}$,
is given by
\[
Q_{ift}=\widetilde{q}_{ft}(\exp\left(\alpha_{i}\right)x_{it}/\exp\left(U_{it}\right),z_{i}).
\]
This is a FELT model, conditional on covariates:
\begin{equation}
Y_{it}=h_{t}(Y_{it}^{*},z_{i}),\label{eq: FELT w z}
\end{equation}
where $h_{t}(Y_{it}^{*},z_{i})=\widetilde{q}_{ft}\left(\exp\left(Y_{it}^{*}\right),z_{i}\right)$,
$Y_{it}^{*}=\alpha_{i}+X_{it}-U_{it}$, and $X_{it}=\ln x_{it}$ is
the logged household budget.

The assumption that the assignable good is normal means that the time-varying
functions $h_{t}$ are strictly monotonic in $Y_{it}^{*}$. One could
additionally impose that the demand functions $\widetilde{q}_{ft}$
come from the application of Roy's Identity to the indirect utility
function (\ref{eq:simplest PIGL}). These demand functions equal a
coefficient times the shadow budget plus a coefficient times the shadow
budget raised to a power, where the coefficients are time-varying
and depend on $z_{i}$.\footnote{Individual demands are derived by the application of Roy's Identity
to (\ref{eq:simplest PIGL}), and are:
\[
q_{jt}(x,z)=c_{jt}\left(z\right)x^{1-r(z)}+b_{t}\left(z\right)x
\]
where $c_{jt}\left(z\right)=-\frac{\nabla_{p}C_{j}(p_{t},z)}{B(p_{t},z)^{r}}$,
$b_{t}\left(z\right)=-\nabla_{p}\ln B(p_{t},z)$. This notation makes
clear that we have time-varying demand functions, due to the fact
that prices vary over time. In our application, prices in each period
are not observed, so we allow the ($z-$dependent) functions $c_{jt}$
and $b_{t}$ to vary over time. We require that the assignable good
be normal, meaning that its demand function is globally increasing
in $x$. This form for demand functions is globally increasing if
$c_{jt}(z)$, $b_{t}(z)$ and $1-r(z)$ are all positive. } We do not impose that additional structure here; instead, we show
in Section \ref{sec:ResourceShares} that the estimated demand curves
given by FELT are close to the PIGL shape restrictions.

Here, the time-dependence of $h_{t}$ is economically important; it
is driven by the price-dependence of preferences and by the fact that
prices are common to all households $i$ but vary over time $t$.
Further, the fixed effects $\alpha_{i}$ are economically meaningful
parameters: they are equal to the logged women's resource shares in
each household. The standard deviation of the logs is a common inequality
measure, and the standard deviation of $\alpha_{i}$ is identified
by FELT given strict monotonicity, as we show in Section \ref{sec:FE}.
Further, the covariation of $\alpha_{i}$ with observed regressors
is identified.%
\begin{comment}
We note that since we cannot identify the location of the distribution
of $\alpha_{i}$, we cannot identify the mean of the distribution
of resource shares. So, without additional information, we could not
identify poverty rates (as done in \citet{dlp13}). Nonetheless, this
model is useful to answer two important questions. First, are fixed
effects (resource shares) fully explained by observed demographics
and budgets, or do we need to appeal to unobserved heterogeneity?
Second, are fixed effects correlated with log-budgets $X_{it}$? Some
results concerning identification of collective household models in
cross-sectional data rely on the assumption of independence.
\end{comment}

\section{Data\label{sec:Data}}

We use data from the $2012$ and $2015$ Bangladesh Integrated Household
Surveys. This data set is a household survey panel conducted jointly
by the International Food Policy Research Institute and the World
Bank. In this survey, a detailed questionnaire was administered to
a sample of rural Bangladeshi households. This data set has two useful
features for our purposes: 1) it includes person-level data on food
intakes and household-level data on total household expenditures;
and 2) it is a panel, following roughly $6000$ households over two
(nonconsecutive) years. The former allows us to use food as the assignable
good to identify our collective household model parameters. The latter
allows us to model household-level unobserved heterogeneity in women's
resource shares.

The questionnaire was initially administered to $6503$ households
in $2012$, drawn from a representative sample frame of all rural
Bangladeshi households. Of these, $6436$ households remained in the
sample in $2015$. In these data, expenditures on food include imputed
expenditure from home production. We drop households with a discrepancy
between people reported present in the household and the personal
food consumption record, and households with no daily food diary data.
Of the remaining data, $6205$ households have total expenditures
reported for both $2012$ and $2015$.

In this paper, we focus on households that do not change members between
periods.\footnote{That is, we exclude households with births, deaths, new members by
marriage or adoption, etc. Although a full-commitment model can accommodate
such changes in household composition, it is easier to think through
the meaning of a person's resource share if the composition is held
constant. } There are $1920$ households whose composition is unchanged between
$2012$ and $2015$. Roughly half of these households have more than
one adult man or more than one adult women. To simplify the interpretation
of estimated resource shares we focus on nuclear households. This
leaves $871$ nuclear households comprised of one man, one woman and
$1$ to $4$ children, where children are defined to be $14$ years
old or younger.

Our household-level annual expenditure, $x_{it},$ is the sum of total
expenditure on, and imputed home-produced consumption of, the following
categories of consumption: rent, food, clothing, footwear, bedding,
non-rent housing expense, medical expenses, education, remittances,
religious food and other offerings (jakat/ fitra/ daan/ sodka/ kurbani/
milad/ other), entertainment, fines and legal expenses, utensils,
furniture, personal items, lights, fuel and lighting energy, personal
care, cleaning, transport and telecommunication, use-value from assets,
and other miscellaneous items. These spending levels derive from one-month
and three-month duration recall data in the questionnaire, and are
grossed up to the annual level. Estimation uses $X_{it}=\ln x_{it}$,
the natural logarithm of annual consumption.

The assignable good, $Y_{it}$, is annual expenditure on food for
the woman. The surveys contains a one-day (24-hour) food diary with
data on \emph{person-level} quantities (measured in kilograms) of
food consumption in 7 categories: Cereals, Pulses, Oils; Vegetables;
Fruits; Proteins; Drinks and Others. These consumption quantities
include home-produced food and purchased food and gifts. They include
both food consumed in the home (both cooked at home and prepared ready-to-eat
food), as well as food consumed outside the home (at food carts or
restaurants). These one-day food quantities are transformed into one-day
food expenditures by multiplying by estimated village-level unit-values
(following \citet{Deaton1997}), and from this, the woman's share
of one-day household food expenditure is calculated. Finally, the
woman's annual expenditure on food is calculated as her share of one-day
food expenditure multiplied by the household's total annual food expenditure
(from recall data).

Our model is also conditioned on a set of time-invariant demographic
variables $z_{i}$. We include several types of observed covariates
in $z_{i}$ that may affect both preferences and resource shares:
1) the age in $2012$ of the adult male; 2) the age in $2012$ of
the adult female; 3) the average age in $2012$ of the children; 4)
the average education in years of the adult male; 5) the average education
in years of the adult female; 6) an indicator that the household has
2 children; 7) an indicator that the household has 3 or 4 children;
and 8) the fraction of children that are girls.\footnote{Since household membership is fixed for all households in our sample,
age, number and gender composition are time-invariant by construction.
However, education level of men and women are time-varying in roughly
20\% of households. For our time-invariant education variables, we
use the average education across the two observed years.} 

For the first five of these demographic variables, in order to reduce
the support of the regressors, we top- and bottom-code each variable
so that values above (below) the $95^{th}$ ($5^{th}$) percentiles
equal the $95^{th}$ ($5^{th}$) percentile values. For all seven
of these variables, we standardize the location and scale so that
their support is $[0,1]$. This support restriction simplifies our
monotonicity restrictions when it comes to estimation, as explained
in Section \ref{sec:ResourceShares}.

We do not trim the data for outliers in the budget or food quantity
demands. Instead, we trim the support of the estimated nonparametric
regression functions to account for fact that these estimators are
high-variance near their boundaries.

Table $2$ in the Appendix, Section \ref{sec:Descriptive-Statistics}
gives summary statistics on these data.

\section{Estimation of Resource Shares \label{sec:ResourceShares}}

Following our identification results, estimation could be based on
composite versions of the maximum score estimator or the conditional
logit estimator (see \citet{BotosaruMuris}). Here, we instead follow
a sieve GMM approach that facilitates the inclusion of a large vector
of demographic conditioning variables $z$ and the imposition of strict
monotonicity on the demand functions (aka: normality of the assignable
good).

The women's food demand equation (\ref{eq: FELT w z}) is a FELT model,
conditional on observed covariates $z$. Denote the inverse demand
functions $g_{t}(Y_{it},z_{it})=h_{t}^{-1}(\cdot,z_{it})$. Given
(\ref{eq: FELT w z}) a two-period setting with $t=1,2$, and time-invariant
demographics $z_{it}=z_{i}$, we have 
\begin{equation}
\alpha_{i}+X_{it}-U_{it}=g_{t}(Y_{it},z_{i}),\label{eq:gt}
\end{equation}
 implying the conditional moment condition
\begin{equation}
E\left[\left.g_{2}(Y_{i2},z_{i})-g_{1}(Y_{i1},z_{i})-\triangle X_{it}\right|X_{i1},X_{i2},z_{i}\right]=0,
\end{equation}
using stationarity of the conditional distribution of the Berkson
errors $\exp\left(-U_{it}\right)$. 

We provide a detailed description of our GMM estimator in Appendix
\ref{app:GMM-Estimator-}. Briefly, we approximate the inverse demand
functions, $g_{t}$, $t=1,2,\text{ }$ using 8th order Bernstein polynomials
to impose monotonicity (estimates for other orders are reported in
Appendix \ref{sec:Additional-estimates}) and use the nonparametric
bootstrap for inference. Because the nonparametric bootstrap may not
be valid for this case, we take our estimated confidence intervals
with caution, see Appendix \ref{app:GMM-Estimator-}.

We characterize several interesting features of the distribution of
resource shares. Recall from Theorems \ref{thm:conditional-mean-a-ID}
and \ref{thm:conditional-var-a-ID} that identification of features
of this distribution does not impose a normalization on assignable
good demand functions, and only identifies the distribution of logged
resource shares (fixed effects) up to location. Consequently, we only
identify features of the resource share distribution up to a scale
normalization. This is related to identification results in \citet{cEkeland09}
which show identification up to location using assignable goods.

Let $\widehat{g}_{it}=\widehat{g}_{t}(Y_{it},z_{i})$ equal the predicted
values of the inverse demand functions at the observed data. Recall
that $g_{t}\left(Y_{it},z_{i}\right)=\alpha_{i}+X_{it}-U_{it},$ so
we can think of $\widehat{g}_{it}-X_{it}$ as a prediction of $\alpha_{i}-U_{it}$.
We then compute the following summary statistics of interest, leaving
the dependence of $\hat{g}_{it},\,t=1,2,$ on $z_{i}$ implicit:
\begin{enumerate}
\item an estimate of the standard deviation of $\alpha_{i}$ given by 
\[
\hat{std}\left(\alpha\right)=\sqrt{\hat{cov}\left(\left(\widehat{g}_{i2}-X_{i2}\right),\left(\widehat{g}_{i1}-X_{i1}\right)\right)},
\]
where $\hat{cov}$ denotes the sample covariance. The standard deviation
of logs is a standard (scale-free) inequality measure. So this gives
a direct measure of inter-household variation in women's resource
shares.
\item an estimate of the standard deviation of the projection error, $e_{i},$
of $\alpha_{i}$ on $\bar{X}_{i}=\frac{1}{2}\left(X_{i1}+X_{i2}\right)$
and $Z_{i}$. Consider the projection
\[
\alpha_{i}=\gamma_{1}\bar{X}_{i}+\gamma_{2}Z_{i}+e_{i},
\]
where $Z_{i}$ contains a constant. We are interested in the standard
deviation of $e_{i}$. To obtain this parameter, we compute estimators
for $\gamma_{1},\gamma_{2}$ from the pooled linear regression of
$\hat{g}_{it}-X_{it}$ on $\bar{X}_{i}\text{ and }Z_{i}.$ Call these
estimators $\hat{\gamma}_{1},\hat{\gamma}_{2}$. Then, as in $\hat{std}\left(\alpha\right)$
in (1), an estimate of the standard deviation of $e_{i}$ is given
by:
\[
\hat{std}\left(e_{i}\right)=\sqrt{\hat{cov}\left(\left(\widehat{g}_{i2}-X_{i2}-\hat{\gamma}_{1}\bar{X}_{i}-\hat{\gamma}_{2}Z_{i}\right),\left(\widehat{g}_{i1}-X_{i1}-\hat{\gamma}_{1}\bar{X}_{i}-\hat{\gamma}_{2}Z_{i}\right)\right)}.
\]
This object measures the amount of variation in $\alpha_{i}$ that
cannot be explained with observed regressors. If it is zero, then
we don't really need to account for household-level unobserved heterogeneity
in resource shares. It is much larger than zero, then accounting for
household-level unobserved heterogeneity is important.
\item an estimate of the standard deviation of $\alpha_{i}+X_{it}$ for
$t=1,2$, computed as
\begin{align*}
\hat{std}\left(\alpha_{i}+X_{it}\right) & =\sqrt{\hat{var}\left(\alpha_{i}\right)+\hat{var}\left(X_{it}\right)+2\hat{cov}\left(\alpha_{i},X_{it}\right)},
\end{align*}
where 
\begin{align*}
\hat{cov}\left(\alpha_{i},X_{i1}\right) & =cov\left(\hat{g}_{i1}-X_{i1},X_{i1}\right),\\
\hat{cov}\left(\alpha_{i},X_{i2}\right) & =cov\left(\hat{g}_{i2}-X_{i2},X_{i2}\right),
\end{align*}
and $\hat{var}\left(\alpha_{i}\right)=\left(\hat{std}\left(\alpha\right)\right)^{2}$
and $\hat{var}\left(X_{it}\right),\,t=1,2,$ is observed in the data.
Since $\alpha_{i}+X_{it}$ is a measure of the woman's shadow budget,
$\hat{std}\left(\alpha_{i}+X_{it}\right)$ is a measure of inter-household
inequality in women's shadow budgets. This inequality measure is directly
comparable to the standard deviation of $X_{i}$ (shown in Table 1),
which measures inequality in household budgets.
\item an estimate of the covariance of $\alpha_{i},X_{it}$ for $t=1,2$,
denoted $\hat{cov}\left(\alpha_{i},X_{it}\right)$. This object is
of direct interest to applied researchers using cross-sectional data
to identify resource shares. If this covariance is non-zero, then
the independence of resource shares and household budgets is cast
into doubt, and identification strategies based on this restriction
are threatened.
\end{enumerate}
Of these, the first 2 summary statistics are about the variance of
fixed effects, and are computed using data from both years. Their
validity requires serial independence of the measurement errors $U_{it}$.
In contrast, the second 2 summary statistics are about the correlation
of fixed effects with the household budget, and are computed at the
year level. They are valid with stationary $U_{it}$, even in the
presence of serial correlation.

We also consider the multivariate relationship between resource shares,
household budgets and demographics. Recall that the fixed effect $\alpha_{i}$
subject to a location normalization; this means that resource shares
are subject to a scale normalization. So, we construct an estimate
of the woman's resource share in each household as $\widehat{\eta}_{i}=\exp\left(\frac{1}{2}\left(\hat{g}_{i1}-X_{i1}\right)+\left(\hat{g}_{i2}-X_{i2}\right)\right)$,
normalized to have an average value of $0.33$. Then, we regress estimated
resource shares $\widehat{\eta}_{i}$ on $\bar{X}_{i}$ and $Z_{i}$,
and present the estimated regression coefficients, which may be directly
compared with similar estimates in the cross-sectional literature. 

The estimated coefficient on $X_{i}$ gives the conditional dependence
of resource shares on household budgets, and therefore speaks to the
reasonableness of the restriction that resource shares are independent
of those budgets (an identifying restriction used in the cross-sectional
literature). Finally, using the estimate of the variance of fixed
effects, we construct an estimate of $R^{2}$ in the regression of
resource shares on observed covariates. This provides an estimate
of how much unobserved heterogeneity matters in the overall variation
of resource shares.

\section{Results}

Figure \ref{fig:pigl_fit_6} shows our estimates of $h_{1}$ and $h_{2}$
(or, equivalently, of $g_{1}$ and $g_{2}$) for $K=8$, for a family
with two children with mean values, $\overline{z}$, of the other
demographics. The figures have food quantities $q_{t}$ on the vertical
axis and $\widehat{g}_{t}\left(q_{t}\right)$ on the horizontal axis,
so the horizontal axis is like a predicted logged household budget.
Solid lines give the nonparametric estimates, and 95\% pointwise confidence
bands for the nonparametric estimates are denoted by dotted lines.
Additionally, to provide reassurance that the PIGL utility model---which
implies the FELT demand curves---fits the data adequately, we display
the PIGL demand curve closest to the FELT estimates in each time period
with dashed lines.\footnote{We compute these PIGL demand curves by nonlinear least squares estimation
of a pooled $q_{it}$ on $\widehat{g}_{it}$, where the demand curves
have the form $q_{it}=c_{t}x^{1-r}+b_{t}x$. We estimate the model
on a grid of 198 points, one for each interior percentile of $q_{it}$
in each period $t=1,2$. } 

Note that since $g_{t}$ are identified only up to location (of $g_{1}$),
we normalize the average of $\widehat{g}_{t}$ to half the geometric
mean of household budgets at $t=$$1$, $\overline{x}_{1}$. Because
estimated nonparametric regression functions can be ill-behaved near
their boundaries, we truncate the estimated functions at the 5th and
95th percentiles of the distribution of $q_{t}$ in each $t$. The
key message from Figure \ref{fig:pigl_fit_6} is that these estimated
demand curves are somewhat nonlinear, estimated reasonably precisely,
and not too far from PIGL. The estimated PIGL curvature parameter
is $r(\overline{z})=0.06$, which means that food demands are close
to PIGLOG (as in \citet{Banks1997}). 
\begin{figure}
\begin{centering}
\includegraphics[scale=0.8]{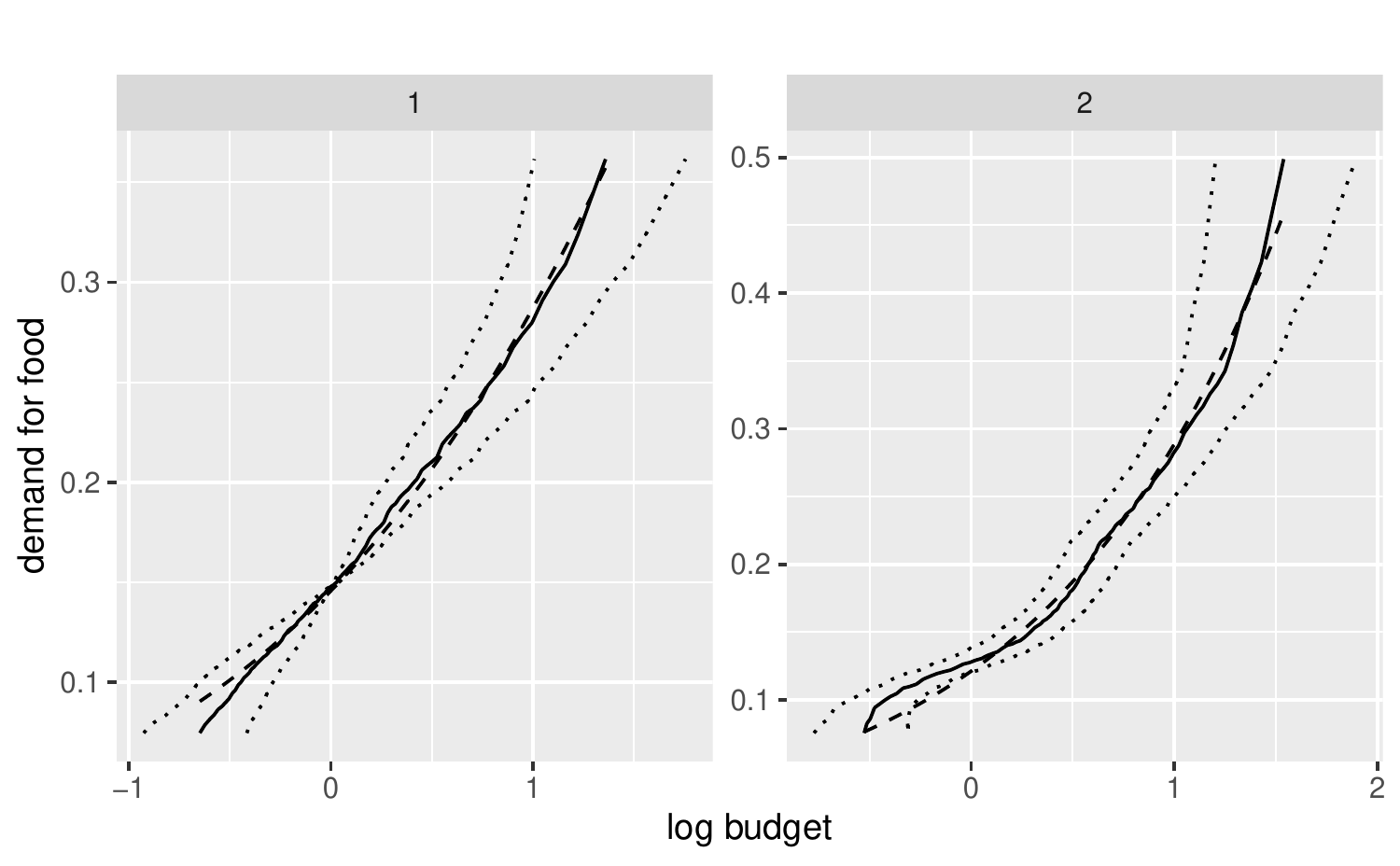}
\par\end{centering}
\label{fig:pigl_fit_6}

\caption{Estimated demand functions. Solid line is the nonparametric estimate,
evaluated at the mean value of the demographics. The dotted lines
indicate the 95\% confidence interval. The dashed line is the PIGL
closest to the nonparametric estimate. Left panel is for period 1,
right panel is for period 2.}
\end{figure}

Table \ref{tab:main_results} gives our summary statistics (items
1-4 above), with bootstrapped 95\% confidence intervals, for our estimates
with 8 Bernstein polynomials (see the Appendix for other lengths of
the Bernstein sieve). In the lower panel, we provide estimated regression
coefficients, also with bootstrapped 95\% confidence intervals, where
we regress estimated resource shares $\widehat{\eta}_{i}$ on log-budgets
$\overline{X}_{i}$ and demographics $z_{i}$.
\begin{table}
\begin{centering}
\begin{tabular}{lccc}
\hline 
 & Estimate & q(2.5) & q(97.5)\tabularnewline
$n=871$ &  &  & \tabularnewline
\hline 
\multicolumn{4}{l}{\textbf{Variability of fixed effects $\alpha_{i}$}}\tabularnewline
$\hat{std}\left(\alpha_{i}\right)$ & 0.2647 & 0.1518 & 0.3707\tabularnewline
$\hat{std}\left(e_{i}\right)$ & 0.1637 & 0.1262 & 0.1931\tabularnewline
$\hat{cov}\left(\alpha_{i},X_{i1}\right)$ & -0.0901 & -0.1292 & -0.0429\tabularnewline
$\hat{cov}\left(\alpha_{i},X_{i2}\right)$ & -0.1034 & -0.1418 & -0.0561\tabularnewline
$\hat{std}\left(\alpha_{i}+X_{i1}\right)$ & 0.3537 & 0.2720 & 0.4605\tabularnewline
$\hat{std}\left(\alpha_{i}+X_{i2}\right)$ & 0.3763 & 0.3010 & 0.4760\tabularnewline
\hline 
\multicolumn{4}{l}{\textbf{Regression estimates}}\tabularnewline
\hline 
$R^{2}$: $\overline{X}_{i},z_{i}$ on $\eta_{i}$ & 0.5205 & 0.3361 & 0.6508\tabularnewline
\hline 
$\overline{X}_{i}$ & -0.0452 & -0.0700 & -0.0196\tabularnewline
age--woman & -0.2513 & -0.4437 & -0.0545\tabularnewline
age--man & 0.2845 & 0.1389 & 0.4191\tabularnewline
2 children & -0.0502 & -0.1694 & 0.0413\tabularnewline
3 or 4 children & -0.1248 & -0.2407 & 0.0140\tabularnewline
avg age of children & -0.0095 & -0.1979 & 0.2151\tabularnewline
fraction girl children & 0.0136 & -0.0881 & 0.1246\tabularnewline
education--woman & -0.0325 & -0.1519 & 0.1198\tabularnewline
education--men & -0.1330 & -0.2589 & -0.0015\tabularnewline
\hline 
\end{tabular}
\par\end{centering}
\caption{Estimates. }
\label{tab:main_results}

\end{table}

Starting with the top panel of Table \ref{tab:main_results}, the
standard deviation of $\alpha_{i}$ is a measure of inter-household
dispersion in women's resource shares. If this dispersion is very
small, then variation in resource shares does not induce much inequality,
and we can reasonably use the household-level income distribution
as a proxy for person-level inequality. However, if the dispersion
is large, then household-level measures of inequality could be very
misleading. Note here that we focus on inequality among women, not
on gender inequality, precisely because the location of resource shares
is not identified.

The estimated value is roughly $0.26$, with a 95\% confidence interval
covering roughly $0.15$ to $0.37$. To get a sense of the magnitude
for the standard deviation of logged resource shares, suppose that
women's resource shares were lognormally distributed. Then our estimated
standard deviation of 0.26 is consistent with 95\% of the distribution
of the resource shares lying in the range $[0.25,0.75]$, which represents
quite a bit of heterogeneity across households.

The next row of Table \ref{tab:main_results} considers how much of
the variation in $\alpha_{i}$ we can explain with observed covariates.
The standard deviation of $e_{i}$ gives a measure of the unexplained
variation, and gives us an idea of whether household-level unobserved
heterogeneity is an important feature of the data. If the standard
deviation of $e_{i}$ is very small, then fixed effects are not needed---conditioning
on observed covariates would be sufficient. Our estimate of the standard
deviation of the unexplained variation in $\alpha_{i}$ is about $0.16$.
This is large relative to the overall estimated standard deviation
of $0.26$, and suggests that accounting for household-level unobserved
heterogeneity is quite important.

The next two rows give the covariance of $\alpha_{i}$ and $X_{it}$.
Here, we see that log resource shares $\alpha_{i}$ strongly and statistically
significantly negatively covary with observed household budgets (the
implied correlation coefficients are close to $-0.8$). This means
that women in poor households are somewhat less poor than they appear
(on the basis of their household budget), and women in richer households
are somewhat more poor than they appear. This is consistent with households
that are closer to subsistence having a more equal distribution of
resources.

The next two rows give the estimated standard deviation of women's
log shadow budgets. This is a scale-free parameter: it does not depend
on the location normalization of $\alpha_{i}$ (which corresponds
to a scale normalization of shadow budgets). The estimated standard
deviations are $0.35$ and $0.38$ in the two periods, respectively.
We can compare these with the standard deviation of log-budgets, reported
in Table 1, of $0.49$ and $0.53$. The point estimates suggest that
there is \emph{less} inequality in women's shadow budgets than in
household budgets. Although the confidence intervals are large, the
test of the hypothesis that the standard deviation of log-budgets
equals the standard deviation of log-shadow budgets rejects in both
years.\footnote{For $H_{0}:Var\left(X_{it}\right)-Var\left(\alpha_{i}+X_{it}\right)$,
we have the following estimated test statistics and (confidence intervals).
Period 1: $0.110\text{}(0.0236,0.165)$; Period 2: $0.137\text{}(0.0527,0.191)$.}

Thus, if we take these results at face value, there is \emph{less}
consumption inequality among women than household-level analysis would
suggest. However, another implication of this is that there is \emph{more}
gender inequality than household level data would suggest. The reason
is that household-level analysis of gender inequality pins gender
inequality on over-representation of one gender in poorer households.
In our data, all households have 1 man and 1 woman, so household-level
analysis of gender inequality would show zero gender inequality. But,
because women in richer households have smaller resource shares, this
induces gender inequality even in these data.

Finding correlation between $\alpha_{i}$ and household budgets is
not sufficient to invalidate previous identification strategies for
cross-sectional settings that rely on independence between resource
shares and household budgets. The reason is that the independence
required is conditional on other observed covariates. To get a handle
on this, the bottom panel of Table \ref{tab:main_results} presents
estimates of coefficients in a linear regression of normalized (to
average $0.33$) estimated resource shares $\widehat{\eta}_{i}$ on
log-budgets $\overline{X}_{i}$ and other covariates $z_{i}$.

The figure below shows the scatterplot of predicted resource shares
versus the log household budget. Here, we see a lot of variation in
resource shares, and it is clearly correlated with household budgets.
The overall variation here provides an estimate of the explained sum
of squares in an infeasible regression of true resource shares $\eta_{i}$
on $\overline{X}_{i}$ and $z_{i}$.\footnote{This artificial regression is infeasible because we observe (through
$g_{it}$) a prediction of $\alpha_{i}+U_{it}$, not of $\alpha_{i}$
itself. However, because we have an estimate of the variance of $\alpha_{i}$,
we can construct an estimate of the variance of $\eta_{i}$ (subject
to the scale normalization that it has a mean of $0.33$). In our
regression, the LHS variable is $\exp\hat{g}_{it}$, which is a prediction
of $\eta_{i}u_{it}$. Since $u_{it}$ are uncorrelated with $\overline{X}_{i}$
and $z$ by assumption, the explained sum of squares from this regression
applies to $\eta_{i}$, and we can use it to form an estimate of $R^{2}$,
which we report, along with a bootstrapped confidence interval. } We may construct an estimate of the total sum of squares of resource
shares from our estimate of the standard deviation of $\alpha_{i}$.
This yields an estimate of $R^{2}$ in the infeasible regression,
which we interpret as the fraction of variation in resource shares
explained by observables.

\begin{figure}
\begin{centering}
\includegraphics[scale=0.5]{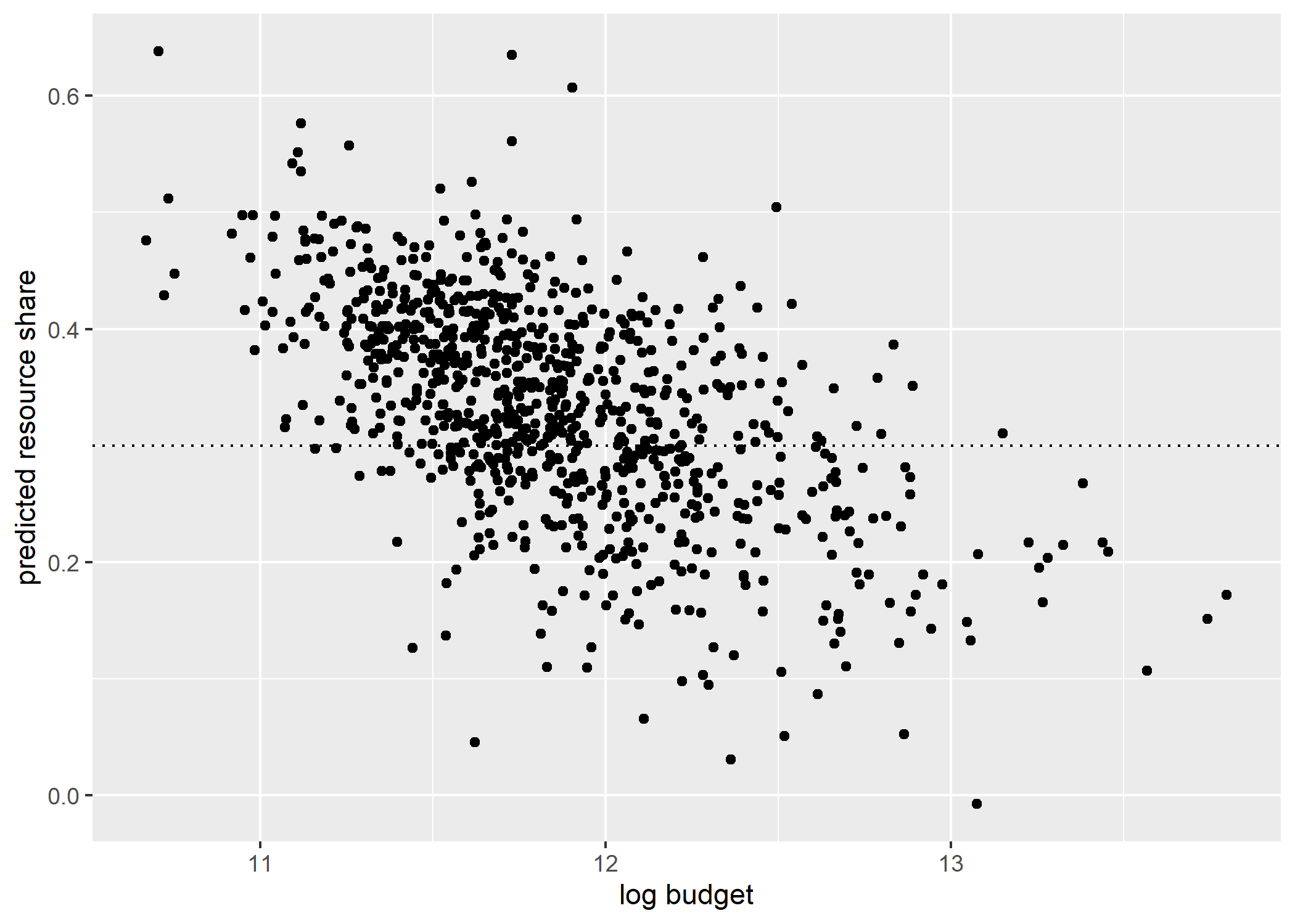}
\par\end{centering}
\caption{Scatterplot of predicted resource shares and log budget.}

\end{figure}

In the first row of the bottom panel, we see that observed variables
explain roughly half the variation in resource shares (the estimate
of $R^{2}$ is $0.52$). This magnitude of explained variation is
very close to that reported in \citet{DunbarLewbelPendakur2019} in
their cross-sectional estimate based on Malawian data. However, whereas
the estimate in \citet{DunbarLewbelPendakur2019} is conditional on
the assumption that unobserved heterogeneity in resources shares is
independent of the household budget, our estimate allows for correlation
of resource shares with the household budget. This large magnitude
of unexplained variation (roughly half) suggests that accounting for
unobserved heterogeneity in resource shares is quite important.

Consider first the coefficient on $\overline{X}_{i}$. The estimated
coefficient is $-0.045$ and is statistically significantly different
from $0$. This means that, even after conditioning on other covariates
(many of which are highly correlated with the budget), we still see
a significant relationship between resource shares and household budgets. 

However, the magnitude of this effect is small. Conditional on $z_{i}$,
the standard deviation of $X_{it}$ is $0.43$ in year 1 and $0.47$
in year 2. Thus, comparing two households with identical $z$ but
which are one standard deviation apart in terms the household budget,
we would expect the woman in the poorer household to have a resource
share $2$ percentage points higher than the woman in the richer household.
Thus, the bulk of the variation that makes the standard deviation
of women's shadow budgets smaller than that of household budgets is
not running through the dependence of resource shares on household
budgets, but rather through the dependence of resource shares on other
covariates that are correlated with household budgets.

We get a very precise estimate of the conditional dependence of resource
shares on household budgets. Overall, then, we see that women's resource
shares are statistically significantly correlated with household budgets,
even conditional on other observed characteristics. But, the estimated
difference in resource shares at different household budgets is quite
small. So, we take this as evidence that the identifying restrictions
used by \citet{dlp13} (and \citet{DunbarLewbelPendakur2019}) may
be false, though perhaps not \emph{very} false. It does suggest that
alternative identifying restrictions---such as those developed here
with a panel model---may be useful.

The rows of Table \ref{tab:main_results} give several other coefficients
that are comparable to other estimates in the literature. \citet{Calvi}
finds that women's resource shares in India decline with the age of
the woman. In these Bangladeshi data, we find evidence that women's
resource shares are strongly negatively correlated with the age of
women and positively correlated with the age of men. 

\citet{dlp13} find that women's resource shares in Malawi decline
with the number of children. Here, we also see that pattern: households
with 2 children have women's resource shares $5$ percentage points
less than households with 1 child; households with 3 or 4 children
have resource shares $12$ percentage points less. \citet{dlp13}
also find that Malawian women's resource shares are higher in households
with girls than households with boys. We do not see evidence of this
in rural Bangladesh: the estimated coefficient on the fraction of
children that are girls statistically insignificantly different from
$0$.

In the Appendix, we also provide estimates analogous to Table 1 using
a different assignable good: clothing. Under the model, using different
assignable goods should yield the same estimates of resource shares.\footnote{Food is a plausible \emph{assignable} good (because if one person
eats it, nobody else can), but it may not be non-shareable (because
there may be scale economies in cooking). In contrast, clothing may
be plausibly non-shareable, but it may not be assignable (because,
e.g., mothers and daughters might wear each others' clothes). See
the Appendix for details on clothing estimates.} This is roughly what we find in our estimates using women's clothing. 

Our estimates use 8th order Bernstein polynomials to approximate the
inverse demand functions\ref{sec:Additional-estimates}. In Appendix
B, we present estimates using Bernstein polynomials of order $K=1,4,8,10$
and show that our finite-dimensional parameter estimates have roughly
the same value for $K\geq8$. 

In summary, in these rural Bangladeshi households, we find evidence
that women's resource shares have substantial dependence on household-level
unobserved heterogeneity and are slightly negatively correlated with
household budgets. The former suggests that random-effects type approaches
to the estimation of resource shares may be inadequate. The latter
suggests that consumption inequality faced by women is actually smaller
than household-level consumption inequality. It also suggests that
cross-sectional identification strategies invoking independence of
resource shares from household budgets, such as \citet{dlp13}, could
be complemented by panel-based identification strategies such as ours.

 \singlespacing\bibliographystyle{apalike}
\bibliography{FELTbib}
\pagebreak{}

\appendix
\textbf{\Large{}ONLINE APPENDICES}{\Large\par}

\section{Proofs\label{sec:Appendix_Proofs}}

\subsection{Proof of Lemma \eqref{lem:Manskimedsign}\label{subsec:Proof-of-Lemma_mediansign}}
\begin{proof}
Define $\overline{D}=1\left\{ D_{1}\left(y_{1}\right)+D_{2}\left(y_{2}\right)=1\right\} $.
The proof consists in showing the following:
\begin{eqnarray}
 &  & \text{med}\left(D_{2}\left(y_{2}\right)-D_{1}\left(y_{1}\right)|X,\overline{\,D}=1\right)\label{eq:M1}\\
 & = & \text{sgn}\left(P\left(D\left(y_{1},y_{2}\right)=\left(0,1\right)|X,\overline{\,D}=1\right)-P\left(D\left(y_{1},y_{2}\right)=\left(1,0\right)|X,\overline{\,D}=1\right)\right)\label{eq:M2}\\
 & = & \text{sgn }\left(\frac{P\left(D\left(y_{1},y_{2}\right)=\left(0,1\right),\overline{\,D}=1|X\right)}{P(\overline{D}=1|X)}-\frac{P\left(D\left(y_{1},y_{2}\right)=\left(1,0\right),\,\overline{D}=1|X\right)}{P(\overline{D}=1|X)}\right)\label{eq:inbetween23}\\
 & = & \text{sgn }\left(P\left(D\left(y_{1},y_{2}\right)=\left(0,1\right),\overline{\,D}=1|X\right)-P\left(D\left(y_{1},y_{2}\right)=\left(1,0\right),\overline{\,D}=1|X\right)\right)\label{eq:alsoinbetween23}\\
 & = & \text{sgn}\left(P\left(D\left(y_{1},y_{2}\right)=\left(0,1\right)|X\right)-P\left(D\left(y_{1},y_{2}\right)=\left(1,0\right)|X\right)\right)\label{eq:M3}\\
 & = & \text{sgn}\left(P\left(D_{2}\left(y_{2}\right)=1|X\right)-P\left(D_{1}\left(y_{1}\right)=1|X\right)\right)\label{eq:M4}\\
 & = & \text{sgn}\left(\Delta X\beta-\gamma\left(y_{1},y_{2}\right)\right)\label{eq:M5}
\end{eqnarray}
where $\left(\ref{eq:M2}\right)$ follows since the random variable
$D_{2}\left(y_{2}\right)-D_{1}\left(y_{1}\right)\in\left\{ -1,1\right\} $,
which implies that 
\begin{eqnarray*}
 &  & \text{med}\left(D_{2}\left(y_{2}\right)-D_{1}\left(y_{1}\right)|X,\overline{\,D}=1\right)\\
 & = & \left\{ \begin{array}{c}
1\text{ if }P\left(D\left(y_{1},y_{2}\right)=\left(0,1\right)|X,\,\overline{D}=1\right)>P\left(D\left(y_{1},y_{2}\right)=\left(1,0\right)|X,\,\overline{D}=1\right)\\
-1\text{ if }P\left(D\left(y_{1},y_{2}\right)=\left(0,1\right)|X,\,\overline{D}=1\right)<P\left(D\left(y_{1},y_{2}\right)=\left(1,0\right)|X,\,\overline{D}=1\right)
\end{array},\right.
\end{eqnarray*}
$\left(\ref{eq:inbetween23}\right)$ follows from the definition of
conditional probability, $\left(\ref{eq:alsoinbetween23}\right)$
follows since the sign function is not affected by scaling both quantities
by the same positive factor (the denominator), $\left(\ref{eq:M3}\right)$
follows by the definition of $\overline{\,D}$, and $\left(\ref{eq:M4}\right)$
follows since: 
\begin{eqnarray*}
P\left(D_{2}\left(y_{2}\right)=1|X\right) & = & P\left(D\left(y_{1},y_{2}\right)=\left(0,1\right)|X\right)+P\left(D\left(y_{1},y_{2}\right)=\left(1,1\right)|X\right)\\
P\left(D_{1}\left(y_{1}\right)=1|X\right) & = & P\left(D\left(y_{1},y_{2}\right)=\left(1,0\right)|X\right)+P\left(D\left(y_{1},y_{2}\right)=\left(1,1\right)|X\right)
\end{eqnarray*}
Finally, $\left(\ref{eq:M5}\right)$ follows from Assumption \ref{A2}(ii),
which implies that, e.g., 
\begin{align*}
P\left(D_{2}\left(y_{2}\right)=1|\alpha,X\right) & >P\left(D_{1}\left(y_{1}\right)=1|\alpha,X\right)\Leftrightarrow\alpha+X_{2}\beta-h_{2}^{-}\left(y_{2}\right)>\alpha+X_{1}\beta-h_{1}^{-}\left(y_{1}\right).
\end{align*}
Integrating both sides over the conditional distribution of $\alpha$
given $X$ obtains: 
\begin{eqnarray*}
P\left(D_{2}\left(y_{2}\right)=1|X\right)>P\left(D_{1}\left(y_{1}\right)=1|X\right) & \Leftrightarrow & X_{2}\beta-h_{2}^{-}\left(y_{2}\right)>X_{1}\beta-h_{1}^{-}\left(y_{1}\right)\\
 & \Leftrightarrow & \Delta X\beta-\gamma\left(y_{1},y_{2}\right)>0.
\end{eqnarray*}
Result $\left(\ref{eq:med_sgn}\right)$ now follows.
\end{proof}

\subsection{Proof of Theorem \ref{thm:ID1}\label{subsec:Proof-of-Theorem_manski_ID}}
\begin{proof}
Following \citet{Manski1985}, it suffices to show that for an arbitrary
$\theta\in\Theta$, $\theta\neq\theta_{0}\equiv\theta_{0}\left(y_{1},y_{2}\right)$,
\begin{equation}
P\left(W\theta<0\leq W\theta_{0}\right)+P\left(W\theta_{0}<0\leq W\theta\right)>0.\label{eq:IDineq}
\end{equation}
Our proof follows very closely that in \citet{Manski1985}, with $W\theta$
taking the role of $xb$ and $W\theta_{0}$ taking the role of $x\beta$.
However, our scale normalization is different.

Without loss of generality, let $X_{K}$ be the continuous regressor
in Assumption \ref{A3}(i). Separate $\Delta X=\left(\Delta X_{-K},\Delta X_{K}\right)$
where the first component $\Delta X_{-K}$ represents all covariates
except the $K$-th one. Similarly, for any $\theta=\left(\beta,\gamma\right)\in\Theta$,
separate $\beta=\left(\beta_{-K},\beta_{K}\right)$. Furthermore denote
$W_{-K}=\left(\Delta X_{-K},-1\right)$ and $\theta_{-K}=\left(\beta_{-K},\gamma\right)$.

Assume that the associated regression coefficient $\beta_{0,K}>0$.
The case $\beta_{0,K}<0$ follows similarly. Let $\theta=\left(\beta,\gamma\right)\in\Theta$,
$\theta\neq\theta_{0}$. As in Manski (1985, p. 318), consider three
cases: (i) $\beta_{K}<0$; (ii) $\beta_{K}=0$; (iii) $\beta_{K}>0$. 

\textbf{Cases (i) and (ii). }$\beta_{K}\leq0$. The proof is identical
to that in Manski (1985), with $X\beta$ replaced by $W\theta$. The
fact that we use a different scale normalization does not come into
play.

\textbf{Case (iii).} $\beta_{K}>0$. note that 
\begin{align*}
P\left(W\theta<0\leq W\theta_{0}\right) & =P\left(-\frac{W_{-K}\theta_{0,-K}}{\beta_{0,K}}<\Delta X_{K}<-\frac{W_{-K}\theta_{-K}}{\beta_{K}}\right).\\
P\left(W\theta_{0}<0\leq W\theta\right) & =P\left(-\frac{W_{-K}\theta_{-K}}{\beta_{K}}<\Delta X_{K}<-\frac{W_{-K}\theta_{0,-K}}{\beta_{0,K}}\right).
\end{align*}
By assumption \ref{ass:normalization}, $\frac{\beta_{-K}}{\beta_{K}}\neq\frac{\beta_{0,-K}}{\beta_{0,K}}$,
which shows that the first $K$ components of the vector $\theta$
are not a scalar multiple of the first $K$ components of the vector
$\theta_{0}.$ Therefore, $\theta$ is not a scalar multiple of $\theta_{0}$.
In particular, $\frac{\theta_{0,-K}}{\beta_{0,K}}\neq\frac{\theta_{-K}}{\beta_{K}}$.
Additionally, assumption \ref{A3}(ii) implies that $P\left(\frac{W_{-K}\theta_{0,-K}}{\beta_{0,K}}\neq\frac{W_{-K}\theta_{-K}}{\beta_{K}}\right)>0$.
Hence at least one of the two probabilities above is positive so that
\eqref{eq:IDineq} holds.
\end{proof}

\subsection{Proof of Theorem \ref{thm:identify-h}\label{subsec:Proof-of-Theorem_identificationfunctions}}
\begin{proof}
Under Assumption \ref{ass:normalizeh}, $h_{1}^{-}(y_{0})=0$. Using
the pair $(y_{0},y_{2})$ for binarization thus obtains identification
of 
\begin{align*}
\gamma\left(y_{0},y_{2}\right) & =h_{2}^{-}(y_{2})-h_{1}^{-}(y_{0})\\
 & =h_{2}^{-}(y_{2}).
\end{align*}
By varying $y_{2}\in\mathcal{\underline{Y}}$, we identify the function
$h_{2}^{-}$ from the binary choice models associated with $\left\{ D\left(y_{0},y_{2}\right)=\left(D_{1}\left(y_{0}\right),D_{2}\left(y_{2}\right)\right),y_{2}\in\underline{\mathcal{Y}}\right\} $.

The pairs $(y_{0},y_{2})$ and $(y_{1},y_{2})$ identify the difference
\begin{align*}
\gamma\left(y_{0},y_{2}\right)-\gamma\left(y_{1},y_{2}\right) & =(h_{2}^{-}(y_{2})-h_{1}^{-}(y_{0}))-(h_{2}^{-}(y_{2})-h_{1}^{-}(y_{1}))\\
 & =h_{1}^{-}(y_{1}).
\end{align*}
By varying $y_{1}\in\mathcal{\underline{Y}}$ we therefore identify
$h_{1}^{-}$.

Thus, the functions $h_{1}^{-}$ and $h_{2}^{-}$ are identified.
Because of monotonicity of $h_{t}$ (Assumption \ref{A1}), and because
$\mathcal{Y}$ is known, $h_{t}^{-}$ contains all the information
about the pre-image of $h_{t}$. Knowledge of the pre-image of a function
is equivalent to knowledge of the function itself. Therefore, $h_{t}$
can be identified from $h_{t}^{-}$.
\end{proof}

\subsection{Proof of Theorem \eqref{thm:logitID}\label{subsec:Proof-of-logitID}\label{subsec:Proof-of-Theorem_identificationlogit}}
\begin{proof}
For the panel data binary choice model with logit errors, we obtain
\begin{align}
 & P(D_{2}\left(y_{2}\right)=1|\overline{D}\left(y_{1},y_{2}\right)=1,X,\alpha)\label{eq:logit_noalpha}\\
 & =\frac{P(D_{2}\left(y_{2}\right)=1,\overline{D}\left(y_{1},y_{2}\right)=1|X,\alpha)}{P\left(\overline{D}\left(y_{1},y_{2}\right)=1|X,\alpha\right)}\label{eq:logit_condo}\\
 & =\frac{P(D_{1}\left(y_{1}\right)=0,D_{2}\left(y_{2}\right)=1|X,\alpha)}{P\left(\overline{D}\left(y_{1},y_{2}\right)=1|X,\alpha\right)}\label{eq:logit_defDbar}\\
 & =\frac{P(D_{1}\left(y_{1}\right)=0,D_{2}\left(y_{2}\right)=1|X,\alpha)}{P(D_{1}\left(y_{1}\right)=0,D_{2}\left(y_{2}\right)=1|X,\alpha)+P(D_{1}\left(y_{1}\right)=1,D_{2}\left(y_{2}\right)=0|X,\alpha)}\label{eq:logit_totalprobs}\\
 & =\frac{1}{1+\frac{P(D_{1}\left(y_{1}\right)=1,D_{2}\left(y_{2}\right)=0|X,\alpha)}{P(D_{1}\left(y_{1}\right)=0,D_{2}\left(y_{2}\right)=1|X,\alpha)}}\label{eq:logit_ratio}\\
 & =\Lambda(\Delta X\beta-\gamma(y_{1},y_{2}))
\end{align}
where \ref{eq:logit_condo} follows from the definition of a conditional
probability; \ref{eq:logit_defDbar} follows because $D_{2}=1$ and
$\bar{D}=1$ are equivalent to $D_{1}=0$ and $D_{2}=1;$ \ref{eq:logit_totalprobs}
follows because $D_{1}+D_{2}=1$ happens precisely when either $\left(D_{1},D_{2}\right)=\left(1,0\right)$
or $\left(D_{1},D_{2}\right)=\left(0,1\right)$; \ref{eq:logit_ratio}
follows by dividing by the numerator; and the final expression follows
by the argument below. 

Note that $\frac{P(D_{1}\left(y_{1}\right)=1,D_{2}\left(y_{2}\right)=0|X,\alpha)}{P(D_{1}\left(y_{1}\right)=0,D_{2}\left(y_{2}\right)=1|X,\alpha)}$
equals

\begin{align}
 & \frac{P(D_{1}\left(y_{1}\right)=1|X,\alpha)P(D_{2}\left(y_{2}\right)=0|X,\alpha)}{P(D_{1}\left(y_{1}\right)=0|X,\alpha)P(D_{2}\left(y_{2}\right)=1|X,\alpha)}\label{eq:logit_serialindependence}\\
 & =\frac{\Lambda\left(\alpha+X_{1}\beta-h_{1}^{-}(y_{1})\right)\left[1-\Lambda\left(\alpha+X_{2}\beta-h_{2}^{-}(y_{2})\right)\right]}{\left[1-\Lambda\left(\alpha+X_{1}\beta-h_{1}^{-}(y_{1})\right)\right]\Lambda\left(\alpha+X_{2}\beta-h_{2}^{-}(y_{2})\right)}\label{eq:logit_probas}\\
 & =\frac{\exp\left(\alpha+X_{1}\beta-h_{1}^{-}(y_{1})\right)}{\exp\left(\alpha+X_{2}\beta-h_{2}^{-}(y_{2})\right)}\label{eq:logit_simplify}\\
 & =\exp\left(\left(X_{1}-X_{2}\right)\beta-\left(h_{1}^{-}(y_{1})-h_{2}^{-}(y_{2})\right)\right),
\end{align}
where \ref{eq:logit_serialindependence} follows from serial independence
of $\left(U_{1},U_{2}\right)$ conditional on $\left(X,\alpha\right)$;
\ref{eq:logit_probas} from the logit model specification; and \ref{eq:logit_simplify}
follows from 
\[
\Lambda\left(u\right)/\left(1-\Lambda\left(u\right)\right)=\exp\left(u\right).
\]

The discussion above implies that \ref{eq:logit_noalpha} does not
depend on $\alpha$. Hence, 
\begin{align*}
p\left(X,y_{1},y_{2}\right) & \equiv P(D_{2}\left(y_{2}\right)=1|\overline{D}\left(y_{1},y_{2}\right)=1,X)\\
 & =\Lambda(\Delta X\beta-\gamma(y_{1},y_{2}))\\
 & =\Lambda\left(W\theta\left(y_{1},y_{2}\right)\right).
\end{align*}
and note that $p\left(X,y_{1},y_{2}\right)$ is identified from the
distribution of $\left(Y,X\right)$, which is assumed to be observed.
Then 
\[
\theta(y_{1},y_{2})=[E(W'W)]^{-1}E(W'\Lambda^{-1}(p(X,y_{1},y_{2})))
\]
by invertibility of $\Lambda$ and the full rank assumption on $E\left[W'W\right]$.
This establishes identification of $\beta$ and $\gamma\left(y_{1},y_{2}\right)$.
The proof in Section \ref{subsec:Proof-of-Theorem_identificationfunctions}
applies, which shows the identification of $h_{1}$ and $h_{2}$.
\end{proof}

\subsection{Proof of Theorems \ref{thm:conditional-mean-a-ID} and \ref{thm:conditional-var-a-ID}\label{subsec:Proof-of-Theorem_conditionalmean}}
\begin{proof}
(a) Note that, without Assumption \ref{ass:normalizeh}, it follows
immediately from the proof of Theorem \ref{thm:identify-h} that we
can only identify $\left\{ h_{t}^{-1}\left(y\right)-c_{1},y\in\mathcal{Y},t=1,2\right\} $,
i.e. we identify the functions $g_{t}\left(y\right)$$,$ $t=1,2.$

Because the functions $g_{1},g_{2}$ are identified, and because the
distribution of $\left(Y,X\right)$ is observable, we can identify
the distribution of the left hand side of the relation below:
\begin{align*}
g_{t}\left(Y_{t}\right)-X_{t}\beta & =h_{t}^{-1}\left(Y_{t}\right)-X_{t}\beta-c_{1}\\
 & =\alpha-U_{t}-c_{1},
\end{align*}
for $t=1,2$.

It follows that for all $x\in\mathcal{X},$
\begin{align*}
\mu\left(x\right) & =E\left[\left.\alpha\right|X=x\right]\\
 & =E\left[\left.\alpha-U_{t}\right|X=x\right]+E\left[\left.U_{t}\right|X=x\right]\\
 & =E\left[\left.h_{t}^{-1}\left(Y_{t}\right)-X_{t}\beta\right|X=x\right]+m\\
 & =E\left[\left.g_{t}\left(Y_{t}\right)-X_{t}\beta\right|X=x\right]+c_{1}+m.
\end{align*}
is identified up to the constants $c_{1}$ and $m$.

The difference in conditional means at any two values $x,x'\in\mathcal{X}$
is therefore identified and given by:
\begin{align*}
 & \mu\left(x\right)-\mu\left(x^{'}\right)=\\
 & \left(E\left[\left.g_{t}\left(Y_{t}\right)-X_{t}\beta\right|X=x\right]+c_{1}+m\right)-\left(E\left[\left.g_{t}\left(Y_{t}\right)-X_{t}\beta\right|X=x^{'}\right]+c_{1}+m\right)\\
 & =E\left[\left.g_{t}\left(Y_{t}\right)-X_{t}\beta\right|X=x\right]-E\left[\left.g_{t}\left(Y_{t}\right)-X_{t}\beta\right|X=x^{'}\right].
\end{align*}

(b) To see that the conditional variance is identified, note that
for all $x\in\mathcal{X},$
\begin{align*}
 & Cov\left(\left.g_{2}\left(Y_{2}\right)-X_{2}\beta,g_{1}\left(Y_{1}\right)-X_{1}\beta\right|X=x\right)=\\
 & =Cov\left(\left.h_{2}^{-1}\left(Y_{2}\right)-X_{2}\beta-c_{1},h_{1}^{-1}\left(Y_{1}\right)-X_{1}\beta-c_{1}\right|X=x\right)\\
 & =Cov\left(\left.\alpha-U_{1}-c_{1},\alpha-U_{2}-c_{1}\right|X=x\right)\\
 & =Var\left(\left.\alpha\right|X=x\right)-Cov\left(\left.\alpha,U_{1}\right|X=x\right)-Cov\left(\left.\alpha,U_{2}\right|X=x\right)+Cov\left(\left.U_{1},U_{2}\right|X=x\right)\\
 & =Var\left(\left.\alpha\right|X=x\right)=\sigma_{\alpha}^{2}\left(x\right),
\end{align*}
where the first equality follows from the definition of $g_{t}$;
the second from the model; the third equality follows from the linearity
of the covariance; and the fourth equality uses assumption (4c) and
(4d).
\end{proof}

\section{GMM Estimator\label{app:GMM-Estimator-}}

\citet{BotosaruMuris} develop $\sqrt{n}$-asymptotically normal estimator
for $\beta$ and (invertible) $h_{t}$ for the FELT model. For the
empirical application in this paper, it is desirable to condition
on household characteristics $z$. This extension is not covered by
\citet{BotosaruMuris}. We describe here a GMM estimator that covers
this case.

The women's food demand equation (\ref{eq: FELT w z}) is a FELT model,
conditional on observed covariates $z$. Denote the inverse demand
functions $g_{t}(Y_{it},z_{it})=h_{t}^{-1}(\cdot,z_{it})$. Given
(\ref{eq: FELT w z}) and a two-period setting with $t=1,2$, and
time-invariant demographics $z_{it}=z_{i}$, we have 
\begin{equation}
\alpha_{i}+X_{it}-U_{it}=g_{t}(Y_{it},z_{i}),\label{eq:gt-1}
\end{equation}
 implying the conditional moment condition
\begin{equation}
E\left[\left.g_{2}(Y_{i2},z_{i})-g_{1}(Y_{i1},z_{i})-\triangle X_{it}\right|X_{i1},X_{i2},z_{i}\right]=0.
\end{equation}

\textit{Sieve estimators}. We approximate the inverse demand functions,
$g_{t}$, $t=1,2,\text{ }$ using Bernstein polynomials. This allows
us to impose monotonicity in a straightforward way, see, e.g., \citet{WangGhosh}.
Let $k=0,...,K$ index univariate Bernstein functions denoted as $\mathcal{B}_{k}\left(\cdot,K\right)$,
where $K$ is the degree of the Bernstein polynomial and where the
Bernstein functions are given by:
\[
\mathcal{B}_{k}\left(u,K\right)=\binom{K}{k}u^{k}\left(1-u\right)^{K-k},\,u\in\left[0,1\right].
\]
Let $l=0,...,L$ index the elements of $z_{i}$, and let the first
(index $0$) element of $z_{i}$ be a constant equal to 1, that is
$z_{i}=\left[1,z_{i1},\dots,z_{iL}\right]$.

Our approximation to $g_{t}\left(Y_{it},z_{i}\right)$ is given by:
\begin{align*}
g_{t}\left(Y_{it},z_{i}\right) & \approx\sum_{k=0}^{K}\beta_{kt}\left(z_{i}\right)\mathcal{B}_{k}\left(Y_{it},K\right)\approx\sum_{l=0}^{L}\sum_{k=0}^{K}z_{i}^{\left(l\right)}\beta_{kt}^{\left(l\right)}\mathcal{B}_{k}\left(Y_{it},K\right).
\end{align*}
For example, when there are no covariates $L=0$, the expression above
reduces to the standard Bernstein polynomial approximation $g_{t}\left(Y_{it}\right)\approx\sum_{k=0}^{K}\beta_{kt}^{\left(0\right)}\mathcal{B}_{k}\left(Y_{it},K\right).$
The Bernstein coefficients are linear functions of the demographics,
and the dependence of the Bernstein coefficients on the demographics
is allowed to vary with time. In this way, the relationship between
the (nonlinear) demand $Y_{it}$ and the latent budget $Y_{it}^{*}$
depends on demographic characteristics and on prices, through $t$.
Since Bernstein polynomials are defined on the unit interval, we normalize
$Y_{it}$ to be uniform on $\left[0,1\right]$ by applying its empirical
distribution function.\footnote{We use this normalization for the estimation of the Bernstein coefficients,
but we present our results in terms of untransformed $Y_{it}$. These
results are obtained by applying the inverse transformation to the
function estimated with transformed data.}

\textit{Unconditional moments}. To form GMM estimators, we construct
the following unconditional moments:
\[
E\left[\left(\sum_{l=0}^{L}\sum_{k=0}^{K}z_{i}^{\left(l\right)}\left(\beta_{k2}^{\left(l\right)}\mathcal{B}_{k}\left(Y_{i2},K\right)-\beta_{k1}^{\left(l\right)}\mathcal{B}_{k}\left(Y_{i1},K\right)\right)-\triangle X_{it}\right)\mathcal{B}_{k'}(x_{it},K)z_{i}^{\left(l'\right)}\right]=0
\]
for $k'=0,\dots,K$, $l'=0,\dots,L$, and $t=1,2$, and

\[
E\left[\left(\sum_{l=0}^{L}\sum_{k=0}^{K}z_{i}^{\left(l\right)}\left(\beta_{k2}^{\left(l\right)}\mathcal{B}_{k}\left(Y_{i2},K\right)-\beta_{k1}^{\left(l\right)}\mathcal{B}_{k}\left(Y_{i1},K\right)\right)-\triangle X_{it}\right)X_{it}z_{i}^{\left(l'\right)}\right]=0,
\]
for $l'=0,\dots,L$ and $t=1,2.$ We include the second condition
(where the logged household budget $X_{it}$ is exogenous) because
we ultimately wish to consider the correlation of $\alpha_{i}$ and
$X_{it}$. For a given order of the sieve approximation, $K$, the
equations above amount to a parametric linear GMM problem.

We impose increasingness on the estimates of the functions $g_{1}$
and $g_{2}$ by imposing $\beta_{kt}\left(z_{i}\right)\geq\beta_{k-1,t}\left(z_{i}\right),$
for all $z_{i}$, for all $t$, and for $k\geq2$. This results in
a quadratic programming problem with linear inequality restrictions,
which we implement in \texttt{R} using the \texttt{quadprog} package.

\textit{Degree of Bernstein polynomial}. Implementing this method
requires the selection of the degree of the Bernstein polynomial,
$K$. While developing a formal selection rule for this parameter
would be desirable, it is beyond the scope of the present paper. Nonetheless,
we adopt an informal selection rule for the number of Bernstein basis
functions -- the smoothing parameter -- based on the following observation.
In our semiparametric setting, the estimators are known to have the
same asymptotic distribution for a range of smoothing parameters (see,
e.g., \citet{Chen2007}). When the number of Bernstein basis functions
is small, the bias dominates and the estimates exhibit a decreasing
bias as the number of terms increases. On the other hand, when the
number of basis functions is large, the statistical noise dominates.
We implement our estimation method over a range of smoothing parameter
values, that is, $K\in\left\{ 1,2,\dots,12\right\} $, in search of
a region where the estimates are not very sensitive to small variations
in the smoothing parameter. We select the mid-point of that region,
so our main results use $K=8$. We present results for $K\in\left\{ 1,4,8,10\right\} $
in the Appendix. In the Appendix table, ``X'' denotes a case where
an estimated variance is negative. 

\textit{Confidence bands}. The confidence bands presented here are
computed via the nonparametric bootstrap. We acknowledge that these
may not be valid as the shape constraints may be binding, see for
example \citet{Andrews99,Andrews00}. Furthermore, they would have
to reflect the choice of the order of the Bernstein polynomial. The
theory to construct valid uniform confidence bands for our estimator
is not yet available. We leave this as a topic for future research.
Two possible solutions are to not impose monotonicity on the estimator
for $h_{t}$ or to use a subset of households with fixed$-$z and
implement the estimator proposed by \citet{BotosaruMuris}. Both monotonicity
and accounting for household characteristics are essential components
of our empirical application, so we do not pursue these approaches.

We use 1,000 bootstrap replications for all our results. We report
pointwise 95\% confidence bands. All reported estimates in the main
text are bias-corrected using the bootstrap.

\textit{Estimates and parameters of interest}. We present estimates
for the demand functions, $\widehat{h}_{t}(x,\overline{z})=\widehat{g}_{t}^{-1}(\cdot,\overline{z})$
for $t=1,2$, where $\bar{z}$ represents a household with 1 child
and has other observed characteristics that are the average of those
of 1-child households. The functions $h_{t}$ (demand functions) are
not of direct interest, but identification of the $h_{t}$ supports
identification of moments of the distribution of fixed effects $\alpha_{i}$.
In our context, $\alpha_{i}$ equals the log of the resource share
of the woman in household $i$. 

We characterize the following interesting features of the distribution
of resource shares. Recall from Theorems \ref{thm:conditional-mean-a-ID}
and \ref{thm:conditional-var-a-ID} that identification of features
of this distribution does not impose a normalization on assignable
good demand functions, and only identifies the distribution of logged
resource shares (fixed effects) up to location. Consequently, we only
identify features of the resource share distribution up to a scale
normalization.

Let $\widehat{g}_{it}=\widehat{g}_{t}(Y_{it},z_{i})$ equal the predicted
values of the inverse demand functions at the observed data. Recall
that $g_{t}\left(Y_{it},z_{i}\right)=\alpha_{i}+X_{it}-U_{it},$ so
we can think of $\widehat{g}_{it}-X_{it}$ as a prediction of $\alpha_{i}-U_{it}$.
We then compute the following summary statistics of interest, leaving
the dependence of $\hat{g}_{it},\,t=1,2,$ on $z_{i}$ implicit:
\begin{enumerate}
\item an estimate of the standard deviation of $\alpha_{i}$ given by 
\[
\hat{std}\left(\alpha\right)=\sqrt{\hat{cov}\left(\left(\widehat{g}_{i2}-X_{i2}\right),\left(\widehat{g}_{i1}-X_{i1}\right)\right)},
\]
where $\hat{cov}$ denotes the sample covariance. The standard deviation
of logs is a standard (scale-free) inequality measure. So this gives
a direct measure of inter-household variation in women's resource
shares.
\item an estimate of the standard deviation of the projection error, $e_{i},$
of $\alpha_{i}$ on $\bar{X}_{i}=\frac{1}{2}\left(X_{i1}+X_{i2}\right)$
and $Z_{i}$. Consider the projection
\[
\alpha_{i}=\gamma_{1}\bar{X}_{i}+\gamma_{2}Z_{i}+e_{i},
\]
where $Z_{i}$ contains a constant. We are interested in the standard
deviation of $e_{i}$. To obtain this parameter, we compute estimators
for $\gamma_{1},\gamma_{2}$ from the pooled linear regression of
$\hat{g}_{it}-X_{it}$ on $\bar{X}_{i}\text{ and }Z_{i}.$ Call these
estimators $\hat{\gamma}_{1},\hat{\gamma}_{2}$. Then, as in $\hat{std}\left(\alpha\right)$
in (1), an estimate of the standard deviation of $e_{i}$ is given
by:
\[
\hat{std}\left(e_{i}\right)=\sqrt{\hat{cov}\left(\left(\widehat{g}_{i2}-X_{i2}-\hat{\gamma}_{1}\bar{X}_{i}-\hat{\gamma}_{2}Z_{i}\right),\left(\widehat{g}_{i1}-X_{i1}-\hat{\gamma}_{1}\bar{X}_{i}-\hat{\gamma}_{2}Z_{i}\right)\right)}.
\]
This object measures the amount of variation in $\alpha_{i}$ that
cannot be explained with observed regressors. If it is zero, then
we don't really need to account for household-level unobserved heterogeneity
in resource shares. It is much larger than zero, then accounting for
household-level unobserved heterogeneity is important.
\item an estimate of the standard deviation of $\alpha_{i}+X_{it}$ for
$t=1,2$, computed as
\begin{align*}
\hat{std}\left(\alpha_{i}+X_{it}\right) & =\sqrt{\hat{var}\left(\alpha_{i}\right)+\hat{var}\left(X_{it}\right)+2\hat{cov}\left(\alpha_{i},X_{it}\right)},
\end{align*}
where 
\begin{align*}
\hat{cov}\left(\alpha_{i},X_{i1}\right) & =cov\left(\hat{g}_{i1}-X_{i1},X_{i1}\right),\\
\hat{cov}\left(\alpha_{i},X_{i2}\right) & =cov\left(\hat{g}_{i2}-X_{i2},X_{i2}\right),
\end{align*}
and $\hat{var}\left(\alpha_{i}\right)=\left(\hat{std}\left(\alpha\right)\right)^{2}$
and $\hat{var}\left(X_{it}\right),\,t=1,2,$ is observed in the data.
Since $\alpha_{i}+X_{it}$ is a measure of the woman's shadow budget,
$\hat{std}\left(\alpha_{i}+X_{it}\right)$ is a measure of inter-household
inequality in women's shadow budgets. This inequality measure is directly
comparable to the standard deviation of $X_{i}$ (shown in Table 1),
which measures inequality in household budgets.
\item an estimate of the covariance of $\alpha_{i},X_{it}$ for $t=1,2$,
denoted $\hat{cov}\left(\alpha_{i},X_{it}\right)$. This object is
of direct interest to applied researchers using cross-sectional data
to identify resource shares. If this covariance is non-zero, then
the independence of resource shares and household budgets is cast
into doubt, and identification strategies based on this restriction
are threatened.
\end{enumerate}
Of these, the first 2 summary statistics are about the variance of
fixed effects, and are computed using data from both years. Their
validity requires serial independence of the measurement errors $U_{it}$.
In contrast, the second 2 summary statistics are about the correlation
of fixed effects with the household budget, and are computed at the
year level. They are valid with stationary $U_{it}$, even in the
presence of serial correlation.

We also consider the multivariate relationship between resource shares,
household budgets and demographics. Recall that the fixed effect $\alpha_{i}$
subject to a location normalization; this means that resource shares
are subject to a scale normalization. So, we construct an estimate
of the woman's resource share in each household as $\widehat{\eta}_{i}=\exp\left(\frac{1}{2}\left(\hat{g}_{i1}-X_{i1}\right)+\left(\hat{g}_{i2}-X_{i2}\right)\right)$,
normalized to have an average value of $0.33$. Then, we regress estimated
resource shares $\widehat{\eta}_{i}$ on $\bar{X}_{i}$ and $Z_{i}$,
and present the estimated regression coefficients, which may be directly
compared with similar estimates in the cross-sectional literature. 

The estimated coefficient on $X_{i}$ gives the conditional dependence
of resource shares on household budgets, and therefore speaks to the
reasonableness of the restriction that resource shares are independent
of those budgets (an identifying restriction used in the cross-sectional
literature). Finally, using the estimate of the variance of fixed
effects, we construct an estimate of $R^{2}$ in the regression of
resource shares on observed covariates. This provides an estimate
of how much unobserved heterogeneity matters in the overall variation
of resource shares.

We also provide estimates that use women's clothing as the assignable
non-shareable good. Here, clothing expenditure is equal to four times
the reported three-month recall expenditure on the following female-specific
clothing items: Saree; Blouse/ petticoat; Salwar kameez; and Orna.
We note that although clothing is a semi-durable good, there are 3
years between the waves of the panel. Consequently, we do not think
that the demands across periods will be strongly correlated due to
the durability of clothing purchased.

Appendix table ``Additional Estimates'' gives results for $K\in\left\{ 1,4,8,10\right\} $
for food (top panel) and clothing (bottom panel). In the table, ``X''
denotes a case where an estimated variance is negative. 

\section{Descriptive Statistics\label{sec:Descriptive-Statistics}}

\begin{table}[H]
\begin{centering}
\begin{tabular}{lllll|ll}
\hline 
\multicolumn{7}{l}{Table 2: Descriptive Statistics}\tabularnewline
\hline 
 & \multicolumn{4}{l|}{raw data} & \multicolumn{2}{l}{top- and bottom-coded}\tabularnewline
$n=871$ &  &  &  &  & \multicolumn{2}{l}{and normalized}\tabularnewline
Variable & Mean & Std Dev & Min & Max & Mean & Std Dev\tabularnewline
\hline 
$X_{i1}$ & 11.68 & 0.49 & 10.24 & 13.68 &  & \tabularnewline
$X_{i2}$ & 12.02 & 0.53 & 10.70 & 14.08 &  & \tabularnewline
$\triangle X_{i}$ & 0.34 & 0.43 & -0.94 & 2.31 &  & \tabularnewline
$Y_{i1}$ & 1561 & 1277 & 0 & 13500 &  & \tabularnewline
$Y_{i2}$ & 1912 & 1727 & 0 & 15000 &  & \tabularnewline
age of woman & 32.76 & 7.64 & 19.00 & 90.00 & 0.35 & 0.25\tabularnewline
age of man & 39.77 & 10.76 & 15.00 & 105.00 & 0.37 & 0.25\tabularnewline
$2$ children & 0.47 & 0.50 & 0.00 & 1.00 & 0.47 & 0.50\tabularnewline
$3$ or $4$ children & 0.19 & 0.39 & 0.00 & 1.00 & 0.19 & 0.39\tabularnewline
age of children & 8.37 & 2.75 & 1.00 & 14.00 & 0.50 & 0.25\tabularnewline
fraction girl children & 0.46 & 0.39 & 0.00 & 1.00 & 0.46 & 0.39\tabularnewline
education of woman & 4.10 & 3.43 & 0.00 & 10.00 & 0.44 & 0.37\tabularnewline
education of man & 3.41 & 3.74 & 0.00 & 10.00 & 0.34 & 0.37\tabularnewline
\hline 
\end{tabular}
\par\end{centering}
\caption{Summary statistics.}
\end{table}

\section{Additional estimates\label{sec:Additional-estimates}}

\begin{sidewaystable}[H]
\begin{centering}
\begin{tabular}{ccccccccccccc}
\hline 
 & {\scriptsize{}1} &  &  & {\scriptsize{}4} &  &  & {\scriptsize{}8} &  &  & {\scriptsize{}10} &  & \tabularnewline
 & {\scriptsize{}Estimate} & {\scriptsize{}q(2.5)} & {\scriptsize{}q(97.5)} & {\scriptsize{}Estimate} & {\scriptsize{}q(2.5)} & {\scriptsize{}q(97.5)} & {\scriptsize{}Estimate} & {\scriptsize{}q(2.5)} & {\scriptsize{}q(97.5)} & {\scriptsize{}Estimate} & {\scriptsize{}q(2.5)} & {\scriptsize{}q(97.5)}\tabularnewline
\hline 
\multicolumn{3}{l}{\textbf{\scriptsize{}Variability in unobservables: Food}} &  &  &  &  &  &  &  &  &  & \tabularnewline
{\scriptsize{}sd\_a} & {\scriptsize{}0.3574} & {\scriptsize{}0.2258} & {\scriptsize{}0.5025} & {\scriptsize{}0.3197} & {\scriptsize{}0.2210} & {\scriptsize{}0.4376} & {\scriptsize{}0.2647} & {\scriptsize{}0.1518} & {\scriptsize{}0.3707} & {\scriptsize{}0.2286} & {\scriptsize{}0.0462} & {\scriptsize{}0.3592}\tabularnewline
{\scriptsize{}sd\_e} & {\scriptsize{}0.1947} & {\scriptsize{}0.1401} & {\scriptsize{}0.2377} & {\scriptsize{}0.1845} & {\scriptsize{}0.1413} & {\scriptsize{}0.2202} & {\scriptsize{}0.1637} & {\scriptsize{}0.1262} & {\scriptsize{}0.1931} & {\scriptsize{}0.1631} & {\scriptsize{}0.1308} & {\scriptsize{}0.1908}\tabularnewline
{\scriptsize{}sd\_ax1} & {\scriptsize{}-0.0506} & {\scriptsize{}-0.0934} & {\scriptsize{}-0.0085} & {\scriptsize{}-0.0805} & {\scriptsize{}-0.1300} & {\scriptsize{}-0.0246} & {\scriptsize{}-0.0901} & {\scriptsize{}-0.1292} & {\scriptsize{}-0.0429} & {\scriptsize{}-0.0781} & {\scriptsize{}-0.1166} & {\scriptsize{}-0.0339}\tabularnewline
{\scriptsize{}sd\_ax2} & {\scriptsize{}-0.0549} & {\scriptsize{}-0.0979} & {\scriptsize{}-0.0129} & {\scriptsize{}-0.0832} & {\scriptsize{}-0.1313} & {\scriptsize{}-0.0277} & {\scriptsize{}-0.1034} & {\scriptsize{}-0.1418} & {\scriptsize{}-0.0561} & {\scriptsize{}-0.0930} & {\scriptsize{}-0.1310} & {\scriptsize{}-0.0497}\tabularnewline
{\scriptsize{}cov\_ax1} & {\scriptsize{}0.5159} & {\scriptsize{}0.4067} & {\scriptsize{}0.6465} & {\scriptsize{}0.4272} & {\scriptsize{}0.3279} & {\scriptsize{}0.5763} & {\scriptsize{}0.3537} & {\scriptsize{}0.2720} & {\scriptsize{}0.4605} & {\scriptsize{}0.3635} & {\scriptsize{}0.2846} & {\scriptsize{}0.4583}\tabularnewline
{\scriptsize{}cov\_ax2} & {\scriptsize{}0.5444} & {\scriptsize{}0.4489} & {\scriptsize{}0.6637} & {\scriptsize{}0.4636} & {\scriptsize{}0.3757} & {\scriptsize{}0.6107} & {\scriptsize{}0.3763} & {\scriptsize{}0.3010} & {\scriptsize{}0.4760} & {\scriptsize{}0.3822} & {\scriptsize{}0.3104} & {\scriptsize{}0.4781}\tabularnewline
\hline 
\multicolumn{3}{l}{\textbf{\scriptsize{}Regression coefficients: Food}} &  &  &  &  &  &  &  &  &  & \tabularnewline
{\scriptsize{}$R^{2}:\overline{X}_{i},z_{i}$ on $\eta_{i}$} & {\scriptsize{}0.6749} & {\scriptsize{}0.5292} & {\scriptsize{}0.8088} & {\scriptsize{}0.6341} & {\scriptsize{}0.4179} & {\scriptsize{}0.8026} & {\scriptsize{}0.5205} & {\scriptsize{}0.3361} & {\scriptsize{}0.6508} & {\scriptsize{}0.5369} & {\scriptsize{}0.3797} & {\scriptsize{}0.6447}\tabularnewline
{\scriptsize{}$\overline{X}_{i}$} & {\scriptsize{}-0.0353} & {\scriptsize{}-0.0643} & {\scriptsize{}-0.0083} & {\scriptsize{}-0.0237} & {\scriptsize{}-0.0534} & {\scriptsize{}0.0039} & {\scriptsize{}-0.0452} & {\scriptsize{}-0.0700} & {\scriptsize{}-0.0196} & {\scriptsize{}-0.0481} & {\scriptsize{}-0.0731} & {\scriptsize{}-0.0234}\tabularnewline
{\scriptsize{}age\_women} & {\scriptsize{}-0.1095} & {\scriptsize{}-0.4028} & {\scriptsize{}0.2182} & {\scriptsize{}-0.2807} & {\scriptsize{}-0.5228} & {\scriptsize{}-0.0498} & {\scriptsize{}-0.2513} & {\scriptsize{}-0.4437} & {\scriptsize{}-0.0545} & {\scriptsize{}-0.3438} & {\scriptsize{}-0.5891} & {\scriptsize{}-0.1424}\tabularnewline
{\scriptsize{}age\_men} & {\scriptsize{}-0.0284} & {\scriptsize{}-0.3205} & {\scriptsize{}0.2123} & {\scriptsize{}0.1146} & {\scriptsize{}-0.0846} & {\scriptsize{}0.2811} & {\scriptsize{}0.2845} & {\scriptsize{}0.1389} & {\scriptsize{}0.4191} & {\scriptsize{}0.2904} & {\scriptsize{}0.1306} & {\scriptsize{}0.4392}\tabularnewline
{\scriptsize{}children2} & {\scriptsize{}-0.1948} & {\scriptsize{}-0.3296} & {\scriptsize{}-0.0860} & {\scriptsize{}-0.1955} & {\scriptsize{}-0.3180} & {\scriptsize{}-0.0884} & {\scriptsize{}-0.0502} & {\scriptsize{}-0.1694} & {\scriptsize{}0.0413} & {\scriptsize{}-0.0821} & {\scriptsize{}-0.2229} & {\scriptsize{}0.0167}\tabularnewline
{\scriptsize{}children3p} & {\scriptsize{}-0.2134} & {\scriptsize{}-0.3372} & {\scriptsize{}-0.1132} & {\scriptsize{}-0.2231} & {\scriptsize{}-0.3440} & {\scriptsize{}-0.0991} & {\scriptsize{}-0.1248} & {\scriptsize{}-0.2407} & {\scriptsize{}0.0140} & {\scriptsize{}-0.1159} & {\scriptsize{}-0.2466} & {\scriptsize{}0.0197}\tabularnewline
{\scriptsize{}age\_children} & {\scriptsize{}-0.3004} & {\scriptsize{}-0.5448} & {\scriptsize{}-0.0620} & {\scriptsize{}-0.0819} & {\scriptsize{}-0.2875} & {\scriptsize{}0.1477} & {\scriptsize{}-0.0095} & {\scriptsize{}-0.1979} & {\scriptsize{}0.2151} & {\scriptsize{}0.0949} & {\scriptsize{}-0.0910} & {\scriptsize{}0.3209}\tabularnewline
{\scriptsize{}frac\_girl} & {\scriptsize{}-0.0338} & {\scriptsize{}-0.1724} & {\scriptsize{}0.1209} & {\scriptsize{}0.0674} & {\scriptsize{}-0.0487} & {\scriptsize{}0.2049} & {\scriptsize{}0.0136} & {\scriptsize{}-0.0881} & {\scriptsize{}0.1246} & {\scriptsize{}-0.0322} & {\scriptsize{}-0.1466} & {\scriptsize{}0.0854}\tabularnewline
{\scriptsize{}edu\_women} & {\scriptsize{}0.0250} & {\scriptsize{}-0.1319} & {\scriptsize{}0.1803} & {\scriptsize{}0.0012} & {\scriptsize{}-0.1546} & {\scriptsize{}0.2036} & {\scriptsize{}-0.0325} & {\scriptsize{}-0.1519} & {\scriptsize{}0.1198} & {\scriptsize{}-0.1096} & {\scriptsize{}-0.2278} & {\scriptsize{}0.0132}\tabularnewline
{\scriptsize{}edu\_men} & {\scriptsize{}0.0079} & {\scriptsize{}-0.1213} & {\scriptsize{}0.1533} & {\scriptsize{}-0.1385} & {\scriptsize{}-0.2764} & {\scriptsize{}0.0091} & {\scriptsize{}-0.1330} & {\scriptsize{}-0.2589} & {\scriptsize{}-0.0015} & {\scriptsize{}-0.0202} & {\scriptsize{}-0.1501} & {\scriptsize{}0.1298}\tabularnewline
\hline 
\multicolumn{3}{c}{\textbf{\scriptsize{}Variability in unobservables: Cloth}} &  &  &  &  &  &  &  &  &  & \tabularnewline
{\scriptsize{}sd\_a} & {\scriptsize{}0.4729} & {\scriptsize{}X} & {\scriptsize{}0.7989} & {\scriptsize{}0.4224} & {\scriptsize{}0.2622} & {\scriptsize{}0.6944} & {\scriptsize{}0.2690} & {\scriptsize{}0.1472} & {\scriptsize{}0.4265} & {\scriptsize{}0.2443} & {\scriptsize{}0.1142} & {\scriptsize{}0.3749}\tabularnewline
{\scriptsize{}sd\_e} & {\scriptsize{}0.2350} & {\scriptsize{}X} & {\scriptsize{}0.3936} & {\scriptsize{}0.2131} & {\scriptsize{}0.1139} & {\scriptsize{}0.2875} & {\scriptsize{}X} & {\scriptsize{}X} & {\scriptsize{}0.1210} & {\scriptsize{}X} & {\scriptsize{}X} & {\scriptsize{}X}\tabularnewline
{\scriptsize{}sd\_ax1} & {\scriptsize{}-0.1254} & {\scriptsize{}-0.2157} & {\scriptsize{}-0.0341} & {\scriptsize{}-0.0503} & {\scriptsize{}-0.1049} & {\scriptsize{}0.0304} & {\scriptsize{}-0.1394} & {\scriptsize{}-0.1745} & {\scriptsize{}-0.0994} & {\scriptsize{}-0.1659} & {\scriptsize{}-0.1964} & {\scriptsize{}-0.1340}\tabularnewline
{\scriptsize{}sd\_ax2} & {\scriptsize{}-0.1185} & {\scriptsize{}-0.2186} & {\scriptsize{}-0.0207} & {\scriptsize{}-0.0903} & {\scriptsize{}-0.1583} & {\scriptsize{}-0.0043} & {\scriptsize{}-0.1832} & {\scriptsize{}-0.2211} & {\scriptsize{}-0.1406} & {\scriptsize{}-0.1903} & {\scriptsize{}-0.2234} & {\scriptsize{}-0.1537}\tabularnewline
{\scriptsize{}cov\_ax1} & {\scriptsize{}0.4633} & {\scriptsize{}X} & {\scriptsize{}0.7828} & {\scriptsize{}0.5667} & {\scriptsize{}0.4144} & {\scriptsize{}0.8386} & {\scriptsize{}0.1927} & {\scriptsize{}X} & {\scriptsize{}0.4568} & {\scriptsize{}X} & {\scriptsize{}X} & {\scriptsize{}0.3213}\tabularnewline
{\scriptsize{}cov\_ax2} & {\scriptsize{}0.5203} & {\scriptsize{}X} & {\scriptsize{}0.8248} & {\scriptsize{}0.5319} & {\scriptsize{}0.3654} & {\scriptsize{}0.8218} & {\scriptsize{}X} & {\scriptsize{}X} & {\scriptsize{}0.4190} & {\scriptsize{}X} & {\scriptsize{}X} & {\scriptsize{}0.3217}\tabularnewline
\multicolumn{3}{c}{\textbf{\scriptsize{}Regression coefficients: Cloth}} &  &  &  &  &  &  &  &  &  & \tabularnewline
{\scriptsize{}$R^{2}:\overline{X}_{i},z_{i}$ on $\eta_{i}$} & {\scriptsize{}11.2428} & {\scriptsize{}10.9458} & {\scriptsize{}14.5224} & {\scriptsize{}0.3753} & {\scriptsize{}-0.0344} & {\scriptsize{}1.0855} & {\scriptsize{}0.8791} & {\scriptsize{}0.6307} & {\scriptsize{}1.2056} & {\scriptsize{}1.1260} & {\scriptsize{}0.9142} & {\scriptsize{}1.3417}\tabularnewline
{\scriptsize{}$\overline{X}_{i}$} & {\scriptsize{}0.1716} & {\scriptsize{}0.0975} & {\scriptsize{}0.2610} & {\scriptsize{}0.0015} & {\scriptsize{}-0.0509} & {\scriptsize{}0.0574} & {\scriptsize{}-0.0925} & {\scriptsize{}-0.1245} & {\scriptsize{}-0.0591} & {\scriptsize{}-0.1488} & {\scriptsize{}-0.1738} & {\scriptsize{}-0.1200}\tabularnewline
{\scriptsize{}age\_women} & {\scriptsize{}0.2658} & {\scriptsize{}-0.4693} & {\scriptsize{}1.1146} & {\scriptsize{}-0.3608} & {\scriptsize{}-0.6580} & {\scriptsize{}-0.0759} & {\scriptsize{}-0.2666} & {\scriptsize{}-0.4136} & {\scriptsize{}-0.0894} & {\scriptsize{}-0.1123} & {\scriptsize{}-0.2224} & {\scriptsize{}0.0056}\tabularnewline
{\scriptsize{}age\_men} & {\scriptsize{}0.1706} & {\scriptsize{}-0.3577} & {\scriptsize{}0.8011} & {\scriptsize{}0.5223} & {\scriptsize{}0.3033} & {\scriptsize{}0.8700} & {\scriptsize{}0.2860} & {\scriptsize{}0.1753} & {\scriptsize{}0.4254} & {\scriptsize{}0.1840} & {\scriptsize{}0.0961} & {\scriptsize{}0.2934}\tabularnewline
{\scriptsize{}children2} & {\scriptsize{}-0.5101} & {\scriptsize{}-0.8971} & {\scriptsize{}-0.1832} & {\scriptsize{}0.0950} & {\scriptsize{}-0.0863} & {\scriptsize{}0.3863} & {\scriptsize{}0.0495} & {\scriptsize{}-0.0380} & {\scriptsize{}0.1572} & {\scriptsize{}0.0623} & {\scriptsize{}0.0008} & {\scriptsize{}0.1369}\tabularnewline
{\scriptsize{}children3p} & {\scriptsize{}-0.3017} & {\scriptsize{}-0.7165} & {\scriptsize{}0.2456} & {\scriptsize{}-0.1110} & {\scriptsize{}-0.2746} & {\scriptsize{}0.1671} & {\scriptsize{}-0.0377} & {\scriptsize{}-0.1448} & {\scriptsize{}0.2903} & {\scriptsize{}-0.0549} & {\scriptsize{}-0.1511} & {\scriptsize{}0.2868}\tabularnewline
{\scriptsize{}age\_children} & {\scriptsize{}-1.1070} & {\scriptsize{}-1.7016} & {\scriptsize{}-0.5342} & {\scriptsize{}-0.0087} & {\scriptsize{}-0.2421} & {\scriptsize{}0.2433} & {\scriptsize{}0.0449} & {\scriptsize{}-0.0820} & {\scriptsize{}0.1722} & {\scriptsize{}-0.0925} & {\scriptsize{}-0.2044} & {\scriptsize{}-0.0050}\tabularnewline
{\scriptsize{}frac\_girl} & {\scriptsize{}-0.0222} & {\scriptsize{}-0.3256} & {\scriptsize{}0.3410} & {\scriptsize{}-0.3711} & {\scriptsize{}-0.5782} & {\scriptsize{}-0.2411} & {\scriptsize{}0.0926} & {\scriptsize{}0.0196} & {\scriptsize{}0.1633} & {\scriptsize{}0.0620} & {\scriptsize{}0.0028} & {\scriptsize{}0.1103}\tabularnewline
{\scriptsize{}edu\_women} & {\scriptsize{}-0.0821} & {\scriptsize{}-0.5592} & {\scriptsize{}0.3867} & {\scriptsize{}-0.0727} & {\scriptsize{}-0.2605} & {\scriptsize{}0.1185} & {\scriptsize{}-0.2170} & {\scriptsize{}-0.3154} & {\scriptsize{}-0.1348} & {\scriptsize{}-0.1029} & {\scriptsize{}-0.1815} & {\scriptsize{}-0.0414}\tabularnewline
{\scriptsize{}edu\_men} & {\scriptsize{}-0.5244} & {\scriptsize{}-0.9658} & {\scriptsize{}-0.1072} & {\scriptsize{}0.0746} & {\scriptsize{}-0.1260} & {\scriptsize{}0.4180} & {\scriptsize{}0.0420} & {\scriptsize{}-0.0497} & {\scriptsize{}0.1742} & {\scriptsize{}0.0167} & {\scriptsize{}-0.0537} & {\scriptsize{}0.1036}\tabularnewline
\end{tabular}
\par\end{centering}
\caption{Summary statistics for log resource shares and their projection on
demographics and budget.}
\end{sidewaystable}

\end{document}